\documentclass[11pt]{article}

\usepackage{amsfonts,amssymb,amsmath,latexsym,ae,aecompl}

\usepackage[noend]{algorithmic}
\usepackage{float}
\usepackage{epsfig}
\usepackage{color}

\newcommand{\F}{\vspace*{\smallskipamount}}

\newcommand{\FFF}{\vspace*{\bigskipamount}}

\newcommand{\BBB}{\vspace*{-\bigskipamount}}

\newcommand{\cA}{{\mathcal A}}

\newcommand{\cF}{{\mathcal F}}

\newcommand{\cO}{{\mathcal O}}
\newcommand{\cP}{{\mathcal P}}

\newcommand{\cS}{{\mathcal S}}
\newcommand{\cT}{{\mathcal T}}

\newcommand{\Paragraph}[1]{\BBB\paragraph{#1}}

\newcommand{\remove}[1]{}
\newcommand{\lbl}{\log(128\,b\log n)}
\newcommand{\llbl}{\log(128b\log n)}

\newcommand{\g}{{\gamma}}

\newcommand{\qed}{\hfill $\square$ \smallskip}

\newenvironment{proof}{\noindent{\bf Proof:}}{\qed}

\def\lpbo {\omega^{\beta/b} }
\def\lpbou {\omega^{1/b} }

\def\lpbi {i^{\beta/b}}

\newlength{\pagewidth}
\newlength{\captionwidth}
\newtheorem{theorem}{Theorem}
\newtheorem{lemma}{Lemma}

\newtheorem{definition}{Definition}
\newtheorem{fact}{Fact}

\begin{document}

\baselineskip 	3ex
\parskip 		1ex

\title{Scalable  Wake-up of Multi-Channel Single-Hop Radio Networks \footnotemark[1] 
				\vfill}

\author{	Bogdan S. Chlebus  	\footnotemark[2] \and
    		Gianluca De Marco \footnotemark[3] \and
		Dariusz R. Kowalski \footnotemark[4]}

\footnotetext[1]{This paper was published in a preliminary form as~\cite{ChlebusDK-OPODIS14} and in its final form as~\cite{ChlebusDK16}.}
		
\footnotetext[2]{Department of Computer Science and Engineering, 
                  University of Colorado Denver, Denver, Colorado 80217, USA.
                  Work supported by the National Science Foundation under Grant No. 1016847.}

\footnotetext[3]{Dipartimento di Informatica,
                  Universit\`a degli Studi di Salerno,
                  Fisciano, 84084 Salerno, Italy.}

\footnotetext[4]{Department of Computer Science,
                	University of Liverpool,
               	Liverpool L69 3BX, United Kingdom.}
		
\date{}

\maketitle

\vfill

\begin{abstract}
We consider  single-hop radio networks with multiple channels as a model of wireless networks. 
There are $n$ stations connected to $b$ radio channels that do not provide collision detection.
A station uses all the channels concurrently and independently. 
Some $k$ stations may become active spontaneously  at arbitrary times. 
The goal is to wake up the network, which occurs when all the stations hear a successful transmission on some channel.
Duration of a waking-up execution is measured starting from the first spontaneous activation.
We present a deterministic algorithm that wakes up a network in $\cO(k\log^{1/b} k\log n)$ time, where $k$ is unknown.
We give a deterministic scalable algorithm for the special case when $b>d \log \log n$, for some constant $d>1$, which wakes up a network in $\cO(\frac{k}{b}\log n\log(b\log n))$ time, with $k$ unknown.
This algorithm misses time optimality by at most a  factor of $\cO(\log n(\log b +\log\log n))$, because any deterministic algorithm requires $\Omega(\frac{k}{b}\log \frac{n}{k})$ time.
We give a randomized algorithm that wakes up a network within $\cO(k^{1/b}\ln \frac{1}{\epsilon})$ rounds with a probability that is at least $1-\epsilon$, for any $0<\epsilon<1$, where $k$ is known.
We also consider a model of jamming, in which each channel in any round may be jammed to prevent a successful transmission, which happens with some known parameter probability~$p$, independently across all channels and rounds.
For this model, we give two deterministic algorithms for unknown~$k$: one wakes up a network in time $\cO(\log^{-1}(\frac{1}{p})\, k\log n\log^{1/b} k)$,  and the other in  time $\cO(\log^{-1}(\frac{1}{p}) \, \frac{k}{b} \log n\log(b\log n))$ when the inequality $b>\log(128b\log n)$ holds, both with probabilities that are at least $1-1/\mbox{poly}(n)$.

\vfill

\end{abstract}

\vfill

\thispagestyle{empty}

\setcounter{page}{0}

\newpage

\section{Introduction}

We consider wireless networks organized as a group of stations connected to a number of  channels.
Each channel provides the functionality of a single-hop radio network.
A station can use any of these channels to communicate directly and concurrently with all the stations.

A restriction often assumed about such networks is that a station can connect to at most one channel at a time for either transmitting or listening.
We depart from this constraint and consider an apparently stronger model in which a station can use all the available channels simultaneously and independently from each other, for instance, some for transmitting and others for listening.
On the other hand, we do not assume collision detection on any channel.

The algorithmic problem we consider is to wake up such a network.
Initially, all the stations are dormant but connected and passively listening to all channels.
Some stations become active spontaneously and want the whole network to be activated and synchronized.
The first successful transmission on any channel suffices to accomplish this goal.

The algorithms we develop are oblivious in the sense that actions of stations are scheduled in advance.
Deterministic oblivious algorithms are determined by decisions for each station when to transmit on each channel and when not. 
Randomized oblivious algorithms are determined by the probabilities for each station and each channel if to transmit on the channel in a round.

We use the following parameters to characterize a multi-channel single-hop radio network.
The number of stations is denoted by $n$ and the number of shared channels by~$b$.
All stations know~$b$.
At most $k$ stations become active spontaneously at arbitrary times and join execution with the goal to wake up the network.
The parameter $k$ is used to characterize scalability of wake-up algorithms, along with the number of channels~$b$.

\Paragraph{Our results.}

We give  randomized and deterministic oblivious algorithms to wake up a multi-channel single-hop radio network.
One of the algorithms scales well with both the number of stations $k$ that may be activated spontaneously and with the number of channels~$b$.

We develop two deterministic algorithms for the case of unknown~$k$.
Our general deterministic algorithm wakes up a network in $\cO(k\log^{1/b} k\log n)$  rounds.
We also give a deterministic algorithm which performs well when sufficiently many channels are available: it wakes up a network in $\cO(\frac{k}{b}\log n\log(b\log n))$ rounds when the numbers of nodes~$n$ and channels~$b$ satisfy the inequality $b>\log(128b\log n)$.
An algorithm of time performance $\cO(\frac{k}{b}\log n\log(b\log n))$, like this one, misses time optimality by at most a  factor of $\log n(\log b+\log\log n)$, because $\Omega(\frac{k}{b}\log \frac{n}{k})$ rounds are required by any deterministic algorithm.
This algorithm is best among those we develop, with respect to scalability with parameters $k$ and~$b$.

We give a randomized algorithm that wakes up a network within $\cO(k^{1/b}\ln \frac{1}{\epsilon})$ rounds with a probability that is at least $1-\epsilon$, for any $0<\epsilon<1$, for  known~$k$.
This algorithm demonstrates a separation between time performance of fastest deterministic algorithms and randomized ones that can use the knowledge of~$k$.

We also consider a model of jamming, in which each channel in any round may be jammed to prevent a successful transmission, which happens with some known probability~$p$, treated as a parameter, independently across all channels and rounds.
For this model, we give two  deterministic algorithms.
One of them wakes up the network in time $\cO(\log^{-1}\frac{1}{p}\, k\log n\log^{1/b} k)$ with a probability that is at least $1-1/\mbox{poly}(n)$.
Another algorithm is designed for the case when the inequality $b>\log(128b\log n)$ holds; in such networks the algorithm operates in  time $\cO(\log^{-1}(\frac{1}{p}) \, \frac{k}{b} \log n\log(b\log n))$ with a large probability.

\Paragraph{Previous work.}

G\k asieniec et al.~\cite{GasieniecPP-JDM01} gave a deterministic oblivious algorithm to wake up a single-hop single-channel radio network in time $\cO(n\log^2 n)$, where $n$ is known and any number of stations may be activated spontaneously.
Our deterministic oblivious algorithms have time performance bounds expressed by formulas in which the three parameters $n$, $b$, and $k$ appear, of which $k$ is unknown while $n$ and $b$ are known.
Observe that if we substitute $k=n$ and $b=1$ in the upper bound $\cO(k\log^{1/b} k\log n)$, which holds for the general deterministic oblivious algorithm, then what is obtained is $\cO(n\log^2 n)$.
Our algorithm, when applied in networks with one channel, has the advantage of scaling with the unknown number $k$ of stations that are activated spontaneously, and provides an asymptotic improvement over the upper bound $\cO(n\log^2 n)$ even for just two channels.

Jurdzi\'nski and Stachowiak~\cite{JurdzinskiS05} gave two randomized algorithms to wake up a multiple access channel.
One of them works in  time $\cO(\log^2 n)$ with high probability, when performance is optimized with respect to~$n$, and another works in time  $\cO(k)$ with high probability, when performance is  optimized with respect to~$k$.
Our randomized algorithm for multi-channel networks has performance sub-linear in~$k$  for even just two channels. 

Koml\'os and Greenberg~\cite{KomlosG-TIT85} showed how to resolve conflict for access to one channel among any of~$k$ stations in time $\cO(k + k\log\frac{n}{k})$, when the stations begin an execution in the same round.
This can be compared to two of our results for the apparently more challenging problem of waking up a network, albeit equipped with multiple channels.
First, our  general deterministic wake-up algorithm runs in   time $\cO(k\log^{1/b} k\log n)$.
Second, when the number of channels satisfies $b=\Omega(\log\log n)$ then another of our algorithms wakes up a network  in time~$\cO(k\log n)$.

\Paragraph{Related work.}

Shi et al.~\cite{ShiHYWL12}  considered the model of a multi-channel network, where there are $n$ nodes connected to $n$ channels, each channel being a single-hop radio network.
A node can use all the available channels concurrently for transmitting and/or receiving transmissions.
They studied the information-exchange problem, in which some $\ell$ nodes start with a rumor each  and the goal is to disseminate all rumors across all stations.
They gave an information-exchange algorithm of time performance $\cO(\log \ell\log\log \ell)$.

The work reported  by Shi et al.~\cite{ShiHYWL12} was the only one, that we are familiar with, to use the model in which nodes can use all the available channels concurrently and independently.
All the other work on algorithms for multi-channel single-hop radio networks used the model in which a node has to choose a channel per round to participate in communication only through this particular channel, either as a listener or as transmitter; variants of this model with adversarial disruptions of channels were also considered.
Next we review work done for this very multi-channel model, in which a station can use at most one channel at a time for communication.

Dolev et al.~\cite{DolevGGN-DISC07} studied a parametrized variant of gossip for multi-channel radio networks.
They gave oblivious deterministic algorithms for an adversarial setting in which a malicious adversary can disrupt one channel per round.
Daum et al.~\cite{DaumGKN12} considered leader election and Dolev et al.~\cite{DolevGGKN-PODC09} gave algorithms to synchronize a network, both papers about an adversarial setting in which the adversary can disrupt a number of channels in each round, this number treated as a parameter in performance bounds.

Information exchange has been investigated extensively for multi-channel wireless networks.
The problem is about some $\ell$ nodes initialized with a rumor each and the goal is either to disseminate the rumors across the whole network or, when the communication environment is prone to failures,  to have each node learn as many rumors as possible.
Gilbert et al.~\cite{GilbertGKN-INFOCOM09} gave a randomized algorithm for the scenario when an adversary can disrupt a number of channels per round, this number being an additional parameter in performance bounds.
Holzer et al. in~\cite{HolzerPSW11} and~\cite{HolzerLPW17} gave deterministic and randomized algorithms to accomplish the information-exchange task in time $\cO(\ell)$, for $\ell$ rumors and for suitable numbers of channels that make this achievable.
This time bound $\cO(\ell)$ is optimal when multiple rumors cannot be combined into compound  messages.
Wang et al.~\cite{WangWYY14} considered information-exchange in a model when rumors can be combined into compound messages and collision detection is available.
They gave an algorithm of time performance $\cO(\ell/b +n\log^2n)$, for $\ell$ rumors and $b$ channels.
Ning et al.~\cite{YuNZJWLF17} gave a randomized algorithm for the model with collision detection that completes information exchange of $\ell$ rumors held by $\ell$ nodes in time $\cO(\ell/b + b\log n)$ with a high probability, where $n$ is not known.

A multi-channel single-hop network is a generalization of a multiple-access channel, which consists of just one channel.
For recent work on algorithms for multiple-access channels see~\cite{AnantharamuC15, AnantharamuCKR-JCSS19,AnantharamuCR-TCS17, AntaMM13,BienkowskiKKK-STACS10, ChlebusKR-DC09, ChlebusKR-TALG12, CzyzowiczGKP11,DeMarcoK15,JurdzinskiS15,Kowalski05}.

The problem of waking up a radio network was first investigated by G\k asieniec et al.~\cite{GasieniecPP-JDM01} in the case of multiple-access channels, see \cite{DeMarcoK17, jDMK13, DeMarcoPS07, JurdzinskiS05} for more on related work.
A broadcast from a synchronized start in a radio network was considered in~\cite{ChlebusGGPR02,ChrobakGR-JA02,ClementiMS03,CzumajR-JA06,DeMarco-soda08,DeMarco-JC10,KowalskiP-DC05}.
The general problem of waking up a multi-hop radio network was studied in~\cite{ChlebusGKR-ICALP05,ChlebusK-PODC04,ChrobakGK-SICOMP07}.

A lower bound for a multiple-access channel was given by Greenberg and Winograd~\cite{GreenbergW-JACM85}.
Lower bounds for multi-hop radio networks we proved by Alon et al.~\cite{AlonBLP91}, Clementi et al.~\cite{ClementiMS03},  Farach-Colton et al.~\cite{Farach-ColtonFM06} and Kushilevitz and Mansour~\cite{KushilevitzM-SICOMP98}.

Ad-hoc multi-hop multi-channel networks were studied by Alonso et al.~\cite{AlonsoKSWW03}, Daum et al. in~\cite{DaumGGKN13} and \cite{DaumKN12}, Dolev et al.~\cite{DolevGKN11}, and So and Vaidya~\cite{SoV04}.

\Paragraph{Structure and history of this document.}

We summarize the technical preliminaries in Section~\ref{sec:technical-preliminaries}.
A lower bound on time performance of deterministic algorithms is given in Section~\ref{sec:lowerbound}.
A randomized wake-up algorithm is given in Section~\ref{sec:randomized-algorithm}.
The concept of a generic deterministic oblivious algorithm is discussed in Section~\ref{sec:deterministic-oblivious-algorithms}.
Instantiations of such a generic algorithm are presented in Section~\ref{sec:general-deterministic-algorithm}, and the next Section~\ref{sec:large} discusses specialized instantiations of the generic deterministic algorithm when the number of channels is sufficiently large with respect to the number of stations.

The results of this paper appeared in a preliminary form in~\cite{ChlebusDK-OPODIS14}.

\section{Technical Preliminaries}

\label{sec:technical-preliminaries}

The model of \emph{multi-channel single-hop radio network} is defined as follows.
There are $n$ nodes attached to a spectrum of $b$ frequencies.
Each frequency determines a multiple access channel.
We use the term ``station'' and ``node'' interchangeably.
The set of all stations is denoted by~$V$, where $|V|=n$.
Each station has a unique name assigned to it, which is an integer in $[1,n]$.

All the available channels operate concurrently and independently from each other.
Each channel has a unique identifier, which is an integer in the interval $[1,b]$.
A station identifies a channel by its identifier, which is the same for all stations.
A station can transmit on any set of channels at any time.
A station obtains the respective feedback from each channel, separately and concurrently among the channels.

\Paragraph{The semantics of channels.}

We say that a station \emph{hears} a message on a channel when the station successfully receives a message transmitted on this channel. 
A channel is \emph{silent} in a time interval when no station transmits on this channel in this time interval.
When more than one stations transmit on a channel, such that their transmissions overlap, then we say that  a \emph{collision} occurs on this channel during the time of overlap.
We say that a channel is equipped with \emph{collision detection} when feedback from the channel allows to distinguish between the channel being silent and a collision occurring on the channel.
When stations receive the same feedback from a channel when it is silent and when a collision occurs on this channel then the channel is said to be \emph{without collision detection}.

When a station transmits on some channel and no collision occurs on this channel during such a transmission  then each station hears the transmitted message on this channel.
When a station transmits a message and a collision occurs during the transmission on this channel  then no station hears this transmitted message.
There could be a collision on one channel and at the same time a message may be heard on some other channel.
There is no collision detection on any channel.

\Paragraph{Synchrony.}

Transmissions on all channels are synchronized.
This means that an execution of an algorithm is partitioned into \emph{rounds}.
Rounds are understood to be of equal length.
Each station has its private clock which is ticking at the rate of rounds. 
Rounds begin and end at the same time for all stations.
When we refer to a round number  then this means the indication  of some station's private clock, while this station is understood from context.

Messages are scaled to duration of rounds, so that transmitting a message takes a whole round.
Two transmissions overlap in time precisely when they are performed in the same round.
This means that two messages result in a collision when and only when they are transmitted on the same channel and in the same round.

\Paragraph{Spontaneous activations.}

Initially, all stations are \emph{passive}, in that they do not execute any communication algorithm, and in particular do not transmit any messages on any channel.
Passive stations listen to all channels all the time, in that when a message is heard on a channel then all passive stations hear it too.

At an arbitrary point in time, some stations become \emph{activated} spontaneously and afterwards they are \emph{active}. 
Passive stations may keep getting activated spontaneously after the round of the first activations.
A specific scenario of timings of certain stations being activated is called an \emph{activation pattern}.
An activated station resets its private clock to zero at the round of activation.
When a station becomes active, it starts from the first round indicated by its private clock to execute a communication algorithm.

Time, as measured by an external observer, is called \emph{global}.
Its units are of the same duration as rounds.
A unit of the global time is called a \emph{time step}. 
The first round of a spontaneous activation of some station is considered as  the first time step of the global time.
The time step in which a station~$u$ becomes activated spontaneously is denoted by~$\sigma_u$.
The set of  stations that are active by time step~$t$ is denoted by~$W(t)$.

\Paragraph{The task of waking up a network.}

The algorithms we consider have as their goal to wake up the network that is executing it.
A network gets \emph{woken up} in the first round when some active station transmits on some channel as the only station transmitting in this round on this particular channel.
This moment is understood as resulting in all passive stations receiving a signal to ``wake up'' and next proceed with executing a predetermined communication algorithm.
The round of waking up a network can be used to synchronize local clocks so that they begin to indicate the same number of a time step.

Time performance of wake-up algorithms is measured as the number of rounds counted from the first spontaneous activation until the round of the first message heard on the network.
Performance bounds of wake-up algorithms in this paper employ the following three parameters: $n$, $b$, and $k$, which are natural numbers such that $1\le k\le n$.
Here $n$ is the number of stations, $b$ is the number of channels, and $k$ denotes an upper bound on the number of stations that may get activated spontaneously in an execution.
Given the parameters $n$, $k$, and $b$, they determine what can be called the \emph{$(n,k,b)$-wake-up problem}: find an algorithm that minimizes the time of waking up a network with $n$ nodes and $b$ channels when up to $k$ stations can be activated spontaneously.

We consider deterministic and randomized algorithms whose goal is to wake up a network.
They are \emph{oblivious} in that the actions of stations are determined in advance; such  a determination is given as the probabilities of actions in the case of randomized algorithms.

A parameter of a system or executions is \emph{known} when it can be used in codes of algorithms.
For an instance of an $(n,k,b)$-wake-up problem, the number of channels $b$ is assumed to be known, which is natural, since stations need to know channels in order to use them.
Regarding the other parameters $n$ and $k$ in this paper, the assumptions are as follows.
If $n$ is known then $k$ is not  assumed as known, which is the case of deterministic algorithms. 
If $k$ is known then $n$ is not assumed to be known, which is the case of a randomized algorithm.

\section{A Lower Bound}

\label{sec:lowerbound}

We present a lower bound on time performance of any deterministic algorithm for the $(n,k,b)$-wake-up problem. 

A family~$\cF$ of subsets of~$[n]$ is said to be \emph{$(n,k)$-selective} when for any subset $A\subseteq [n]$ of $k$ elements there exists a set $b\in\cF$ such that $A\cap B$ is a singleton set.
There is a straightforward correspondence between $(n,k)$-selective families and deterministic oblivious wake-up protocols on a multiple-access channel with $n$ stations when up to $k$ stations are activated spontaneously.

Clementi et al.~\cite{ClementiMS03} showed that $\Omega(k\log\frac{n}{k})$ is a lower bound  on time needed to wake-up  a single-channel single-hop radio network with $n$ nodes, when some $k$ nodes are activated spontaneously.
More precisely, Clementi et al.~\cite{ClementiMS03} showed that any $(n,k)$-selective family needs to have at least $\frac{k}{24}\lg\frac{n}{k}$ elements, for $k$ such that $2 < k\le\frac{n}{64}$. 

Wake-up protocols for our model of multi-channel networks can also be interpreted as $(n,k)$-selective families.
An additional aspect is that we can apply $b$ sets from the family simultaneously as concurrent transmissions on different channels.
This directly implies that $\frac{1}{b}\cdot \frac{k}{24}\lg\frac{n}{k}$ time is required of a wake-up protocol for a $b$-channel network of $n$ nodes, for $k$ such that $2< k\le \frac{n}{64}$, by a lower bound on the size of selective families given in~\cite{ClementiMS03}.

In the remaining part of this section, we demonstrate a lower bound on time of wake-up for multi-channel networks.
The arguments we expound follow the main ideas of the proof of a lower bound given by Clementi et al.~\cite{ClementiMS03} for one channel; in particular, we also refer to properties of intersection-free families proved by Frankl and~F{\"u}redi~\cite{FranklF85}.
There are two goals for including this Section, rather than simply accepting $\frac{1}{b}\cdot \frac{k}{24}\lg\frac{n}{k}$ as a lower bound.
First, proving Theorem~\ref{thm:lower-bound} makes the paper self-contained.
Second, we state the lower bound in Theorem~\ref{thm:lower-bound} in a form that, first, improves the key involved constants, by obtaining $\frac{1}{b}\cdot \frac{k}{4}\lg\frac{n}{k}+ \cO(\frac{k}{b})$ instead of $\frac{1}{b}\cdot \frac{k}{24}\lg\frac{n}{k}$,  and, second,  relaxes the restriction $k\le \frac{n}{64}$ to a general case $k\le n$.

We define a \emph{query} to be a set of ordered pairs $(x,\beta)$, for $x \in V$ and $1\leq \beta\leq b$. 
An interpretation of a pair $(x,\beta)\in Q$, for a query $Q$, is that station $x$ is to transmit on channel~$\beta$ at the time step assigned for the query.  
In this section, a deterministic oblivious algorithm $\cA$ is represented as a sequence of queries $\cA = \{Q_1,\ldots, Q_t\}$. 
The index $i$ of a query $Q_i$ in such a sequence $\cA$ is interpreted as the time step assigned for the query.
We use the notation 
\[
Q_{i,\beta} = \{ x\in V : (x,\beta) \in Q_i \}
\ , 
\]
for a query~$Q_i$. 
This represents the subset of all stations that transmit on channel~$\beta$ in time step~$i$.

We use the Iverson's bracket $[\cP]$, where  $\cP$ is a logical statement, that could be either true or false, defined as follows: $[\cP]=1$ if $\cP$ is true and $[\cP]=0$ if $\cP$ is false.
We denote by $\cF_{k}^n$ the family of sets with exactly~$k$ elements out of $n$ possible elements, interpreted as $k$-sets of stations taken from among all $n$ stations.

\begin{lemma}
\label{pigeonhole}

Let $\cA =\{Q_1,Q_2,\ldots,Q_t\}$ be a sequence of queries representing an algorithm. 
There exists a  sub-family $\cS \subseteq \cF_{k}^n$ with at least $|\cF_{k}^n|/2^{bt}$ elements such that any two sets $A$ and $B$ in~$\cS$ satisfy 
\begin{equation}
\label{eqn:equivalence}
[|A\cap Q_{i,\beta}| \text{\rm\   is odd}\,] = [|B\cap Q_{i,\beta}| \text{\rm\  is odd}\,] 
\ ,
\end{equation}
for all $i$ and $\beta$ such that $1\le i \le t$ and $1\le \beta\le b$.
\end{lemma}

\begin{proof}
Two sets $A$ and $B$ in $\cF_{k}^n$ are said to be \emph{$i$-similar} when the equality~\eqref{eqn:equivalence} holds for all $\beta$ such that $1\le \beta\le b$.
The relation of $i$-similarity is an equivalence relation on~$\cF_{k}^n$.

The proof is by induction on~$t$.
The base of induction is obtained by taking an equivalence class of $1$-similarity that is of a largest size.
This size is at least  $|\cF_{k}^n|/2^{b}$, by the pigeonhole principle.

For the inductive step, assume that the claim holds for $i$ such that $0\le i \le t$, that is, there exists a  sub-family $\cS \subseteq \cF_{k}^n$ with at least $|\cF_{k}^n|/2^{bt}$ elements such that any two sets $A$ and $B$ in~$\cS$ satisfy the identity~\eqref{eqn:equivalence}, for all $i$ and $\beta$ such that $1\le i \le t$ and $1\le \beta\le b$.
Consider the relation of $(t+1)$-similarity determined on~$\cS$ by a query~$Q_{t+1}$.
There are at most $2^b$ nonempty equivalence classes of this relation.
One of them has at least $ |\cS(i)|/2^b$ elements, by the pigeonhole principle. 
By the inductive assumption, the size of this equivalence class is at least $|\cF_{k}^n|/2^{b(t+1)}$.
\end{proof}

For $\lambda \leq \kappa \leq n$, a family $\cF \subseteq \cF_{\kappa}^n$ is said to be 
\emph{$(n, \kappa, \lambda)$-intersection free} if $|F_1 \cap F_2 | \not= \lambda$ for every $F_1$ and $F_2$ in $ \cF_{\kappa}^n$.

\begin{fact}[\cite{FranklF85}]
\label{fact:FF}
For any $(n, \kappa, \lambda)$-intersection free family $\cF$  the following inequality holds:
\[
 |\cF| 
 \leq 
 {n \choose \lambda} \cdot
 \frac{{ 2\kappa-\lambda-1 \choose \kappa}}{{2\kappa -\lambda -1 \choose   \lambda}}
 \ ,
\]
assuming the inequality $2\lambda + 1 \geq \kappa$  and that $\kappa - \lambda$ is a prime power.
\hfill $\square$
\end{fact}

\begin{lemma}
\label{lem:binomial-simplification}
The following identity holds true
\[
 {3k/2-1 \choose k} \Big/ {3k/2-1 \choose k/2} = \frac{1}{2}
\ .
\]
for  integers $k\ge 0$ and $n\ge 0$ such that $k$ is even and $k\le 2n$.
\end{lemma}

\begin{proof}
Let $k=2m$.
It is sufficient to verify the following equation:
\[
2\cdot \binom{3m-1}{ 2m} = \binom{3m-1 }{ m}
\ .
\]
This  indeed is the case, as the following transformations
\begin{eqnarray*}
\binom{3m-1 }{ m} &=& \binom{3m-1}{2m-1}\\
&=& \frac{(3m-1)!}{m! \cdot (2m-1)!}\\
&=& \frac{2\cdot (3m-1)!}{(m-1)!\cdot (2m)!}\\
&=& 2\cdot\binom{3m-1}{2m} 
\end{eqnarray*}
provide the needed verification.
\end{proof}

We use notation $\lg x$ to denote the binary logarithm $\log_2 x$.

\begin{lemma}
\label{l:sets}

Let $\cA =\{Q_1,Q_2,\ldots,Q_t\}$ be an algorithm, where the following inequality holds
\[
t\leq \frac{k}{2b}\lg \frac{n}{k} - \frac{3k-2}{2b}
\] 
and $\frac{k}{2}$ is a prime power. 
There exist two sets $A,B\subseteq \cF_k^n$ such that the following are satisfied:
\begin{enumerate}
\item[\rm (a)] $|A \cap B| = \frac{k}{2}$, 
\item[\rm (b)] $[|A\cap Q_{i,\beta}|  \text{\rm\  is odd}\,] = [|B\cap Q_{i,\beta}|  \text{\rm\   is odd}\,]$, 
for every $i$ and $\beta$ such that $1\le i \le t$ and $1\le \beta\le b$.
\end{enumerate}
\end{lemma}    
\begin{proof}
By Lemma~\ref{pigeonhole}, there exists a sub-family $\cS \subseteq \cF_{k}^n$ of $|\cS|$ elements in $\cF_{k}^n$ and such that 
\begin{equation}\label{S}
 |\cS|\ge |\cF_{k}^n|/2^{bt} = {n \choose {k}} / 2^{bt}
\end{equation}
and
\[
[|A  \cap Q_{i,\beta}|  \mbox{ is odd}] = [|B\cap Q_{i,\beta}|  \mbox{ is odd}]
\ ,
\] 
for every $A,B\in\cS$, $1\le \beta\le b$ and $1\le i \le t$.
It follows that any two sets $A$ and $B$ in $\cS \subseteq \cF_{k}^n$ satisfy condition~(b).

It remains to demonstrate that there are at least two sets in $\cS$ that also satisfy condition~(a), that is, their intersection has $\frac{k}{2}$ elements.
We use Fact~\ref{fact:FF}, for $\kappa=k$ and $\lambda=\frac{k}{2}$, such that $\frac{k}{2}$ is a prime power.
It gives that any sub-family of $\cF_{k}^n$ containing sets that have pairwise intersections of size 
different from $k/2$ has  at most these many elements:
\[
{n \choose k/2} \cdot {3k/2-1 \choose k} \Big/ {3k/2-1 \choose k/2} 
\ ,
\]
which equals $\frac{1}{2}{n \choose k/2}$ by Lemma~\ref{lem:binomial-simplification}.
It follows that it is sufficient for the following inequality to hold:
\begin{equation}
\label{eqn:half-of-binom}
 |\cS| > \frac{1}{2} {n \choose k/2}
 \ .
\end{equation}
To demonstrate this, we start from inequality~\eqref{S} and proceed through a sequence of inequalities. 
In the process, we use the following estimates on binomial coefficients, for positive integers $x$:
\[
\left(\frac{n}{x}\right)^x \le \binom{n }{x} < \left(\frac{ne}{x}\right)^x
\ ,
\]
along with the assumed bound on $t$ in the form $bt\le \frac{k}{2}\lg \frac{n}{k}-\frac{3k-2}{2}$.
The algebraic manipulations are as follows:
\begin{eqnarray*}
  |\cS|  &\geq&  \binom{n}{k}/2^{bt} \\
         &\geq&  2^{k\lg (n/k) - bt}  \mbox{  }\\
         &\geq&  2^{k\lg (n/k) - (k/2)\lg (n/k) + (3k-2)/2} \;\;\mbox{\small after substituting the bound on $bt$}\\
         &=&     2^{(k/2)\lg (n/k) + 3k/2-1} \\
                  &=&     2^{(k/2)\lg (8n/k) -1} \\
         &>&  2^{(k/2)\lg (2ne/k)-1}  \\
         &=&     \frac{1}{2}  \left(\frac{2ne}{k}\right)^{k/2}  \\
         &>&  \frac{1}{2}   \binom{n }{ k/2}  \ .
\end{eqnarray*}
We have thus justified~\eqref{eqn:half-of-binom}.
This in turn implies that  there exist  two sets in $\cS$ whose intersection has exactly  $\frac{k}{2}$ elements.
This completes the proof of existence of two sets $A$ and $B$ in $\cS$ that satisfy part~(a). 
\end{proof}

Now we proceed to prove the lower bound, which is formulated as follows.


\begin{theorem}
\label{thm:lower-bound}

Any deterministic oblivious algorithm that wakes up a network of $n$ nodes with $b$ channels, when at most $k$ nodes are activated spontaneously, for $2< k\le n$, requires more than  
\[
\frac{k}{4b}\lg \frac{n}{k} - \frac{3k-2}{2b}
\]
 time steps.
\end{theorem}

\begin{proof}
Let $i$ be the largest integer such that $2< 2^i \leq k$. 
Assume that $k' = 2^i$ stations, out of at most $k$ available stations, are activated simultaneously at time step zero.
Let $\cA =\{Q_1,Q_2,\ldots,Q_t\}$ be an algorithm such that the following inequality holds:
\begin{equation}
\label{eqn:lower-bound}
t \leq \frac{k'}{2b}\lg \frac{n}{k'} - \frac{3k'-2}{2b}
\ .
\end{equation}
Lemma~\ref{l:sets} is applicable, because $k'/2$ is power of $2$ and so a prime power.
Let $A$ and $B$ be two subsets of $\cF_{k'}^n$, with the properties as stated in Lemma~\ref{l:sets}.
Let us set $A'=A\setminus B$ and $B'=B \setminus A$. 
Observe that if $A$ and~$B$ have properties  (a) and~(b) of  Lemma~\ref{l:sets} then the following holds for $A'$ and $B'$:
\begin{enumerate}
\item[\rm (a*)] $|A'| = |B'| = \frac{k'}{2}$,

\item[\rm (b*)] $A'\cap B'=\emptyset$, 

\item[\rm (c*)] $[|A'\cap Q_{i,\beta}|  \mbox{ is odd}]  = [|B'\cap Q_{i,\beta}|  \mbox{ is odd}]$, 
for every $1\le \beta\le b$ and $1\le i \le t$.
\end{enumerate}
We set $X = A' \cup B'$ to obtain that (a*) and (b*) imply $|X| = k'$. 
Moreover, from (c*) it follows that 
either $X\cap Q_{i,\beta} = \emptyset$
or $|X\cap Q_{i,\beta}| \ge 2$, for all $1\le \beta\le b$ and $1\le i \le t$.
Consider an execution in which the stations in $X$ are simultaneously activated as the only stations activated spontaneously.
Then, during the first $t$ time steps after activations, no station in $X$ is heard on any channel.

We conclude that if an algorithm $\cA$ always wakes up the network within $t$ steps then the 
inequality~\eqref{eqn:lower-bound} cannot hold.
Consequently, since $k/2<k'\le k$, the lower bound follows.
\end{proof}

\section{A Randomized Algorithm}

\label{sec:randomized-algorithm}

A pseudocode of a randomized algorithm, called \textsc{Channel-Screening}, is in Figure~\ref{alg:channel-screening}.
All random bits generated during an execution are independent from each other.
The pseudocode refers to~$k$, which means it is known.
At the same time, $n$ needs not to be known, because only active stations participate in the execution, so their number is always bounded above by~$k$.


\begin{figure}[t]
\rule{\textwidth}{0.75pt}

\F 
\textbf{Algorithm} \textsc{Channel-Screening}

\rule{\textwidth}{0.75pt}
\begin{center}
\begin{minipage}{\pagewidth}
\begin{description}
\item[\texttt{repeat}]  \texttt{concurrently} for each channel $\beta$ 
\begin{description}
\item[\rm transmit] a message on channel $\beta$ with probability  $k^{-\beta/b}$
\end{description}
\item[\tt until] a message is heard on some channel
\end{description}
\end{minipage}
\FFF

\rule{\textwidth}{0.75pt}

\parbox{\captionwidth}{\caption{\label{alg:channel-screening}
A randomized algorithm. 
The same pseudocode is used by any station  that gets activated spontaneously. 
}}
\end{center}
\end{figure}

\begin{lemma} 
\label{lem:class} 

Let $t$ be a time step and let $1\leq \beta\leq b$ be such that the following inequalities hold:
\[
k^{(\beta-1)/b} \leq |W(t)| \leq k^{\beta/b}
\ .
\] 
When algorithm \textsc{Channel-Screening} is executed then the probability of hearing a message at time step~$t$ on channel $\beta$ is at least $\frac{1}{2ek^{1/b}}$.
\end{lemma}

\begin{proof}
Let $S(\beta, t)$ be the event of a successful transmission on channel $\beta$ at time $t$.
The probability that a station $w\in W(t)$ transmits at time $t$ on channel $\beta$ while all the
others remain silent is
\begin{eqnarray*}
       \Pr (S(\beta, t) ) &\geq& \frac{|W(t)|}{k^{\beta/b}} \left( 1- \frac{1}{k^{\beta/b}} \right)^{|W(t)|-1}\\
                         &\geq& \frac{k^{(\beta-1)/b}}{k^{\beta/b}} \left( 1- \frac{1}{k^{\beta/b}} \right)^{ k^{\beta/b}},
\end{eqnarray*}
where the last inequality follows from the hypothesis that $k^{(\beta-1)/b} \leq |W(t)| \leq k^{\beta/b}$.
Hence
\[
     \Pr (S(\beta, t) ) \geq \frac{1}{2ek^{1/b}}\ ,
\]
which completes the proof.
\end{proof}

An  estimate the number of rounds needed  to make the probability of failure smaller than a threshold $\epsilon$ is given in the following theorem:


\begin{theorem}
\label{thm:randomized-algorithm}

Algorithm \textsc{Channel-Screening}  executed by $n$ nodes on a network with $b$ channels and when at most $k$  stations get activated succeeds in waking up the network  in $\cO(k^{1/b}\ln \frac{1}{\epsilon})$ time with a probability that is at least $1-\epsilon$.
\end{theorem}

\begin{proof}
Let us consider a set of contiguous time steps $T$.
For $1\leq \beta \leq b$, let us use the following notation:
\[
T_{\beta} = \{ t \in T |\ k^{(\beta-1)/b} \leq |W(t)| \leq k^{\beta/b}\}\ .
\]
Let $\bar{E}(t)$ be the event of an unsuccessful time step $t$, in which no station transmits as the only transmitter on any channel, and let $\bar{E}(\beta, t)$ be the event of an unsuccessful time step $t$ on channel~$\beta$, with $1\leq \beta \leq b$.
By Lemma~\ref{lem:class}, the probability of having a sequence of $\lambda = |T|$ unsuccessful 
time steps can be estimated as follows:
\begin{eqnarray*}
       \Pr\Bigl ( \bigcap_{t \in T} \bar{E}(t) \Bigr )
             & \leq &
                         \Pr \Bigl ( \bigcap_{t \in T_1} \bar{E}(1, t) \Bigr ) \cdot
                         \Pr \Bigl ( \bigcap_{t \in T_2} \bar{E}(2, t) \Bigr ) \cdots
                         \Pr \Bigl ( \bigcap_{t \in T_b} \bar{E}(b, t) \Bigr ) \\
             & \leq &
                         \Bigl( 1 - \frac{1}{2ek^{1/b}} \Bigr)^{|T_1|} \cdot
                         \Bigl( 1 - \frac{1}{2ek^{1/b}} \Bigr)^{|T_2|} \cdots
                         \Bigl( 1 - \frac{1}{2ek^{1/b}} \Bigr)^{|T_b|} \\ 
             & \leq &   
                         \Bigl( 1 - \frac{1}{2ek^{1/b}} \Bigr)^{\lambda} \\
             & \leq & 
                         \epsilon,                                 
\end{eqnarray*}
for $\lambda \geq 2ek^{1/b}\ln \frac{1}{\epsilon}$.
\end{proof}

Algorithm \textsc{Channel-Screening} demonstrates that the lower bound of Theorem~\ref{thm:lower-bound} can be beaten by a randomized algorithm that can use the magnitude of the parameter~$k$ explicitly.
Actually, the bound of Theorem~\ref{thm:randomized-algorithm} is such that just for $b=2$ channels the network is woken up with a positive probability in time that is $\cO(\sqrt{k})$ while the lower bound of Theorem~\ref{thm:lower-bound} implies  that time~$\Omega(k)$ is required.

\section{A Generic Deterministic Oblivious Algorithm}

\label{sec:deterministic-oblivious-algorithms}

Deterministic oblivious algorithms are represented as schedules of transmission precomputed for each station.
A schedule is simply a binary sequence.
Such schedules of transmission are organized as rows of a binary matrix, for the sake of visualization and discussion.

Let $\ell$ be a positive integer treated as a parameter.
An array~$\cT$ of entries of the form $T(u,\beta,j)$, where $u\in V$ is a station, $\beta$ such that $1\leq \beta \leq b$ is a channel,  and integer $j$ is such that $0\leq j\leq \ell$, is a \emph{transmission array} when each entry is  either a~$0$ or a~$1$.
The parameter $\ell=\ell(\cT)$ is called the \emph{length} of array~$\cT$. 
Entries of a transmission array~$\cT$ are called \emph{transmission bits} of~$\cT$.
The number~$j$ is the \emph{position} of a transmission bit~$T(u,\beta,j)$.

For a transmission array $\cT$, a station~$u$ and channel~$\beta$, the sequence of entries $T(u,\beta,j)$,  for $j=1,\ldots,\ell$, is called a \emph{$(u,\beta)$-schedule} and is denoted $\cT(u,\beta,\ast)$.
A $(u,\beta)$-schedule $\cT(u,\beta,\ast)$ defines the following \emph{schedule of transmissions} for station~$u$: transmit on  channel~$\beta$ in the $j$th round if and only if $T(u,\beta,j) = 1$.
Every station $u\in V$ is provided with a copy of all entries  $T(u,\ast,\ast)$ of some transmission array~$\cT$ as a way to instantiate the code of a wake-up algorithm.

A pseudocode representation of such a generic deterministic oblivious algorithm  is given in Figure~\ref{alg:wake-up}.
In analysis of performance, based on properties of~$\cT$, it is assumed that the number~$n$ is known and the parameter~$k$ is unknown.


\begin{figure}[t]
\rule{\textwidth}{0.75pt}

\F 
\textbf{Algorithm} \textsc{Wake-Up}\,$(\cT)$

\rule{\textwidth}{0.75pt}
\begin{center}
\begin{minipage}{\pagewidth}
\begin{description}
\item[\texttt{concurrently for}] each channel $\beta$ 
\begin{description}
\item[\texttt{for}]  $j\gets 1$ \texttt{to} $\ell(\cT)$ \texttt{do}
\begin{description}
\item[\texttt{if}] $T(u,\beta,j) = 1$ \texttt{then} transmit on  channel~$\beta$ in round $j$ 
\item[\tt if] a transmission was just heard on some channel \texttt{then exit}
\end{description}
\end{description}
\end{description}
\end{minipage}
\FFF

\rule{\textwidth}{0.75pt}

\parbox{\captionwidth}{\caption{\label{alg:wake-up}
A pseudocode for a station $u$ that gets activated spontaneously of a generic algorithm. 
The algorithm is instantiated by  a transmission array~$\cT$, where $\ell(\cT)$ is its length.
Station~$u$ knows only the part $\cT(u,\ast,\ast)$ of~$\cT$.
}}
\end{center}
\end{figure}

Time is measured by each station's private clock, with rounds counted from the activation.
Let us recall that if a station $u$ is active in a time step~$t$ then $u$ perceives this time step~$t$ as round~$t-\sigma_u$.

A station $v$ is $\beta$-{\em isolated} at time step~$t$ when $v \in W(t)$ and when both $T(v,\beta, t-\sigma_v) = 1$ and $T(u,\beta, t-\sigma_u) = 0$, for every $u\in W(t) \setminus \{v\}$.
A station $v$ is \emph{isolated at time step $t$} when $v$ is $\beta$-isolated at time step~$t$ for some channel $1\leq \beta \leq b$.

For a given transmission array,  by an \emph{isolated position} we understand a pair $(t,\beta)$ of time step~$t$ and channel~$\beta$ such that there is a $\beta$-isolated station at time step~$t$.
When $(t,\beta)$ is an isolated position of a transmission array~$\cT$  then a successful wake-up occurs by time~$t$ when algorithm \textsc{Wake-Up}\,$(\cT)$ is executed.

We organize a transmission array by partitioning it into sections of increasing length.
We will use the property of the mapping $i\mapsto 2^i\cdot i^{1/b}$ to be strictly increasing, which can be verified directly.

Let $c$ be a positive integer and let us define the function $\varphi$ as follows:
\begin{enumerate}
\item
$\varphi (0) = 0$, and 
\item
$\varphi(i) = c2^i\cdot i^{1/b} \lg n$, for positive integers~$i$.
\end{enumerate}
The $i$th \emph{section} of a $(u,\beta)$-schedule $T(u,\beta,*)$, for $1\leq i \leq \lg n$,  consists of a concatenation of all the segments 
\[
T(u,\beta, \varphi(i)), T(u,\beta, \varphi(i)+1), \ldots, T(u,\beta, \varphi(i+1)-1)
\] 
of transmission bits. 
A station executing the $i$th section of its schedules is said to be  \emph{in stage~$i$}.
The stations that are in a stage $i$ at a time step~$j$ are denoted by~$W_i(j)$.
The constant $c$ will be determined later as needed.

The  identity 
\[
\bigcup_{i=1}^{\lg n} W_i(j) = W(j)
\] 
holds for every time step $j$, because an active station is in some stage.
Observe that the length of the $i$th section for any $(u,\beta)$-schedule is $\varphi(i+1) - \varphi(i)$, which can be verified to be at least as large as~$\varphi(i)$.

These time steps at which sufficiently many stations are in a stage, say, $\omega$, and no station is involved in a stage with index larger than $\omega$, play a special role in the analysis.
The relevant notions are that of a balanced time step, given in  Definition~\ref{bal}, and also of a balanced time interval, given in Definition~\ref{baltime}.
Intuitively, a balanced time step is a round at which there are  some $\Theta(2^{\omega})$ awaken stations involved in section $\omega$ and no station involved  in the subsequent sections, for some $1\leq \omega \leq \log n$. 
Similarly, a balanced time interval is a time interval that includes sufficiently many balanced time slots.
These notions are precisely defined as follows.

\begin{definition}[Balanced time steps.]
\label{bal}
For a stage $\omega$, where $1\leq \omega \leq \log k$, a time step~$j$ is \emph{$\omega$-balanced} when the following properties hold: 
\begin{enumerate}
\item[\rm (a)] ${2^{\omega}}\leq |W_{\omega}(j)| \leq {2^{\omega+2}}$, and
 \item[\rm (b)]  $|W_{i}(j)| = 0$, for all stages $i$ such that $i > \omega$.
 \end{enumerate}
\end{definition}

When we refer to time intervals then this means intervals of time steps of the global time.
We identify  time intervals with sufficiently large sequences of consecutive time steps that 
contain only balanced time steps as follows:

\begin{definition}[Balanced time intervals.]
\label{baltime}

Let $\omega$ be a stage,  where $1\leq \omega \leq \log k$.
A time interval $[t_1, t_2]$ of size $\varphi(\omega-1)$, 
is said to be $\omega$-\emph{balanced}, if every time step $j \in [t_1, t_2]$ is $\omega$-balanced.
An interval is called  \emph{balanced} when there exists a stage $\omega$, for $1\leq \omega\leq \log k$,  such that it is $\omega$-balanced.
\end{definition}

For a time step $j$, we define $\Psi (j)$ as follows:
\[
\Psi (j) = \sum_{\omega=1}^{\log k} \frac{|W_\omega(j)|}{2^i}.
\] 
We categorize balanced time intervals further by considering their useful properties:

\begin{definition}[Light time intervals.]
\label{d:light}
Let $\omega$ be a stage,  where $1\leq \omega \leq \log k$.
An $\omega$-balanced time interval $[t_1, t_2]$ is called \emph{$\omega$-light} when  
\begin{enumerate}
\item[\rm (1)] the inequality $\Bigl| \bigcup_{i=1}^{\omega} W_{i}(j) \Bigr| \leq 2^{\omega+4}$ 
holds for every time step $j\in [t_1,t_2]$, and 
\item[\rm (2)] interval $[t_1, t_2]$ contains at least $\varphi(\omega-2)$ time steps $j$ such that  
\begin{equation}
\label{condition}
   1 \leq \Psi(j) \leq 128 \cdot \omega
 \ .
\end{equation} 
\end{enumerate}
An interval is called  \emph{light} when there exists a stage $\omega$, for $1\leq \omega\leq \log n$,  such that it is $\omega$-light.
\end{definition}

We show existence of a suitable array of waking schedules by the probabilistic method.
In the argument, we refer to randomized transmission arrays, as defined next.
These are arrays with entries being independent random variables.

\begin{definition}[Regular randomized transmission arrays.]
\label{def:array}

A \emph{randomized transmission array~$\cT$} has the structure of a transmission array. Transmission bits $T(u,\beta,j)$ are not fixed but instead are independent Bernoulli random variables.
Let  $u$ be a station and  $\beta$ denote a channel.
For $1\leq i\leq \lg n$, the entries of the $i$th section of the  $(u,\beta)$-schedule are stipulated to have the following probability distribution
\[
\Pr (T(u,\beta,j) = 1 ) = 2^{-i}\cdot i^{-\beta/b}
\ ,
\] 
for $j= \varphi(i), \ldots, \varphi(i+1)-1$.
\end{definition}

We say that the number of channels $b$ is \emph{$n$-large}, or simply \emph{large}, or that there are \emph{$n$-many channels}, when the inequality $b>\lbl$ holds.
We set 
\[
\varphi(i)=c\cdot (2^i/b) \lg n\lbl
\] 
for such $b$, where $c$ is a sufficiently large constant to be specified later.
Recall the notation 
\[
\Psi (j) = \sum_{i=1}^{\log k} \frac{|W_i(j)|}{2^i}
\ ,
\]
for a time step~$j$. 
For $n$-many channels, we use a modified version of a light time interval (see Definition~\ref{d:light}), where  condition~(2) is replaced by the following one:
\begin{equation}
\label{newcondition}
   1 \leq \Psi(j) \leq 128 \cdot \log n.
\end{equation} 
For a channel~$\beta$, we use the following notation:
\begin{equation}
\label{eqn:beta-star}
\beta^*=\beta \bmod \lbl
\ .
\end{equation}

A \emph{modified randomized transmission array~$\cT$} has the structure of a transmission array. 
Transmission bits $T(u,\beta,j)$ are not fixed but instead are independent Bernoulli random variables.
Let  $u$ be a station and  $\beta$ denote a channel.
For $1\leq i\leq \lg n$, the entries of the $i$th section of the  $(u,\beta)$-schedule are stipulated to have the following probability distribution
\[
\Pr (T(u,\beta,j) = 1 ) = b \cdot 2^{-i -\beta^*}
\ ,
\] 
for $j= \varphi(i), \ldots, \varphi(i+1)-1$, where we use the notation stipulated in~\eqref{eqn:beta-star}.

A randomized transmission array, whether regular or modified, is used to represent a randomized wake-up algorithm.
To decide if a station~$u$ transmits on a channel~$\beta$ in a round~$j$, this station first carries out a Bernoulli trial with the probability of success as stipulated in the definition of the respective randomized array, and transmits when the experiment results in a success.
Regular arrays are used in the general case and modified arrays when there are $n$-many channels.

\begin{definition}[Waking arrays.]
\label{schedule}

A transmission array $\cT$ is said to be \emph{waking} when for every $k$ such that $1\leq k\leq n$ and a light interval $[t_1, t_2]$ such that $|W(t)| \leq k$, whenever $t_1 \leq t \leq t_2$,
there exist both a time step~$j\in [t_1, t_2]$ and a station $w\in W(j)$ such that $w$ is isolated 
at time step~$j$. 
\end{definition}

The length of a waking array  is the worst-case time bound on performance of the wake-up algorithm determined by this transmission array.

\section{A General Deterministic Algorithm}

\label{sec:general-deterministic-algorithm}

We consider waking arrays $\cT$ to be used as instantiations of algorithm \textsc{Wake-Up} given in Section~\ref{sec:deterministic-oblivious-algorithms}.
The goal is to minimize their length in terms of $n$ while their effectiveness is to be expressed in terms of both $n$ and $k\le n$.
The existence of waking arrays of small length is shown by the probabilistic method.
The main fact proved in this Section is as follows:


\begin{theorem}
\label{thm:general-deterministic}

There exists a deterministic waking array $\cT$ of length $\cO(n\log n \log^{1/b} n)$ such that, when used to instantiate the generic algorithm \textsc{Wake-Up}, produces an algorithm \textsc{Wake-Up}\,$(\cT)$ that wakes up a network in  time $\cO(k\log n\log^{1/b} k)$, for up to $k\le n$ stations  activated spontaneously.
\end{theorem}

We proceed with a sequence of preparatory Lemmas.
Let $X$ be the set of stations that are activated first.
Let $\sigma_X$ be the time step at which they become active. 
All stations are passive before the time step~$\sigma_X$.

Let $\g_0 = 0$ and for $i = 1,2,\ldots,\lg n$, define $\g_i$ as the sum of the lengths of the first $i$ sections. 
We have the following identities 
\[
\g_i = \sum_{h=1}^{i} (\varphi(h+1) - \varphi(h)) = \varphi(i+1)
\ .
\]

\begin{lemma}
\label{zigzag}

For $i = 1,2,\ldots,\lg n$, all stations in $X$ are in section $i$ of their transmission 
schedules between time $\sigma_X + \g_{i-1}$ and time $\sigma_X + \g_i-1$.
\end{lemma}

\begin{proof}
Any station $x \in X$, woken up at time $\sigma_X$,  for $1 \leq i \leq \lg n$, reaches section $i$ at time $\sigma_X + \g_{i-1}$ and continues to transmit according to transmission bits in section $i$ until time $\sigma_X + \g_i-1$. 
\end{proof}

\begin{lemma}
\label{join}

Fix a time step $j'$ and an integer $\omega$, with $1 \leq \omega \leq \log n$. 
For any integer $h \geq 1$, there exists a time step $j''\geq j'$ such that the following holds for
$j = j'', \ldots, j'' + \varphi(\omega+h)$:
\[
     \bigcup_{i=1}^{\omega} W_i(j') \subseteq  W_{\omega+h+1}(j) 
     \ .
\] 
\end{lemma}

\begin{proof}
Let us fix $h\geq 1$. 
Recall that the sum of the lengths of the first $i$ sections is $\g_i = \varphi(i+1)$.
Any station $x\in W_1(j')$ is in section $\omega+h+1$ by time $j'+\varphi(\omega+h+1)$.
Analogously, a station $y\in W_{\omega}(j')$ cannot leave section $\omega+h+1$ 
before time step $j' + \varphi(\omega+h+2) - \varphi(\omega+1)$.
It follows that all stations in $W_i(j')$, for $1\leq i\leq \omega$, are in section $\omega+h+1$ 
between the time step 
\[
j'' = j'+\varphi(\omega+h+1)
\] 
and the time step
\[
\tau = j' + \varphi(\omega+h+2) - \varphi(\omega+1)
\ .
\]
It remains to count the number of time  steps between~$j''$ and~$\tau$.
We have  that
\begin{eqnarray*}
       \tau - j'' &=& \varphi(\omega+h+2) - \varphi(\omega+h+1) - \varphi(\omega+1) \\
                  &\geq& \varphi(\omega+h+1) - \varphi(\omega+1) \\
                  &\geq& \varphi(\omega+h),
\end{eqnarray*}
for every $h\geq 1$.
\end{proof}

\begin{lemma}
\label{fat}

Fix a time step $j'$ and an integer $\omega$, with $1 \leq \omega \leq \log n$.
Assume the following two inequalities:
\[
 \Bigl| \bigcup_{i=1}^{\omega-1} W_{i}(j') \Bigr| \geq 3\cdot |W_{\omega}(j')| \ \ \ \mbox{and}\ \ \ 
|W_{\omega}(j')| \geq 2^{\omega}
 \ .
\]
Then there exists an interval $[t_1, t_2]$ of size $\varphi(\omega+1)$ with $t_1 \geq j'$ such that
$|W_{\omega+2}(j)| \geq 2^{\omega+2}$. 
\end{lemma}

\begin{proof}
We have the following estimate:
\begin{eqnarray}
    \Bigl| \bigcup_{i=1}^{\omega} W_i(j') \Bigr| &=&  \Bigl| \bigcup_{i=1}^{\omega-1} W_i(j') \Bigr| + |W_\omega(j')|\nonumber \\
                                               &\geq& 3\cdot |W_\omega(j')| + |W_\omega(j')|\nonumber \\
                                               &=&   2^{\omega+2}. \label{und}
\end{eqnarray}
By Lemma ~\ref{join}, there exists a round $j''\geq j'$ such that for
$j = j'', \ldots, j'' + \varphi (\omega+1)$, 
\[
     \bigcup_{i=1}^{\omega} W_i(j') \subseteq  W_{\omega+2}(j).
\] 
Therefore, for $j = j'', \ldots, j'' + \varphi(\omega+1)$,  the following bounds hold:
\[
    \Bigl|W_{\omega+2}(j)\Bigr| \geq \Bigl| \bigcup_{i=1}^{\omega} W_i(j') \Bigr| \geq {2^{\omega+2}}
    \ ,
\]
where the last step follows from~\eqref{und}.
We conclude by setting $t_1 =  j''$ and $t_2 = j'' + \varphi(\omega+1)$.
\end{proof}

\begin{lemma}
\label{iter}

Suppose that $[t_1, t_2]$ is an interval of size $\varphi(\omega-1)$, for some $1 \leq \omega\leq \log k$, 
such that for every round $j \in [t_1, t_2]$, the following conditions hold:
\begin{itemize}
  \item [\rm (a')]  $|W_{\omega}(j)| \geq {2^{\omega}}$;
  \item [\rm (b')] for $i>\omega$, $|W_{i}(j)| = 0$.
\end{itemize}
Then there exists an $\omega'$-balanced interval for some $\omega \leq \omega'\leq \log k$.
\end{lemma}

\begin{proof}
If $|W_{\omega}(j)| \leq {2^{\omega+2}}$ for every $j\in [t_1, t_2]$ and condition (a) of Definition~\ref{bal} holds then there is nothing to prove. 
Therefore, assume that there exists $j'\in [t_1, t_2]$ such that $|W_{\omega}(j')| > {2^{\omega+2}}$.
Observe that since at most $k$ stations can be activated, we must have $\omega < \log k - 2$.
Let $h\geq 1$ be an integer such that the following inequalities hold:
\[
2^{\omega+h+1} < |W_{\omega}(j')| \leq 2^{\omega+h+3}
\ .
\]
By Lemma~\ref{join}, there exists a round $j''\geq j'$ such that, for
$j = j'', \ldots, j'' + \varphi (\omega+h)$, the inclusion
\[
     \bigcup_{i=1}^{\omega} W_{i}(j') \subseteq  W_{\omega+h+1} (j)
\] 
holds.
Therefore $|W_{\omega+h+1} (j)| \geq 2^{\omega+h+1}$ for $j = j'', \ldots, j'' + \varphi (\omega+h)$.
Let $t'_1 = j''$, $t'_2 = j'' + \varphi (\omega+h)$ and $\omega' = \omega+h+1$. We have found an interval
$[t'_1, t'_2]$ of size $\varphi (\omega'-1)$ such that for every round $j \in [t'_1, t'_2]$, 
the following conditions hold:
\begin{itemize}
  \item [1.]  $|W_{\omega'}(j)| \geq {2^{\omega'}}$;
  \item [2.] for $i>\omega'$, $|W_{i}(j)| = 0$.
\end{itemize}
If $|W_{\omega'} (j)| \leq 2^{\omega'+2}$ for every $j \in [t'_1, t'_2]$ then interval $[t'_1, t'_2]$ is $\omega'$-balanced and we are done; otherwise we repeat the same reasoning to find a new interval.
Since the number of stations that can be woken up is bounded by $k$, there must exist
an interval $[\tau_1, \tau_2]$ of size $\varphi (\iota-1)$, for some $1 \leq \iota \leq \log k$, such that
$|W_{\iota} (j)| \leq 2^{\iota+2}$ for every $j \in [\tau_1, \tau_2]$.
\end{proof}

\begin{lemma}
\label{proper}

There exists an $\omega$-balanced interval $[t_1, t_2]$,
for some $1\leq \omega \leq \log k$, such that 
\begin{equation}\label{up}
 \Bigl| \bigcup_{i=1}^{\omega-1} W_{i}(j) \Bigr| < 3\cdot |W_{\omega}(j)| 
\end{equation}
for every $j \in [t_1,t_2]$. 
\end{lemma}

\begin{proof}
Let us pick $\iota$ such that $1\leq \iota \leq \log k$ and $2^{\iota}\leq |X| \leq 2^{\iota+1}$. 
By Lemma~\ref{zigzag}, all stations in $X$ are in section~$\iota$, for  $j\in [\sigma_X + \g_{\iota-1}, \sigma_X + \g_{\iota}-1]$. 
Therefore the inequality $|W_{\iota}(j)| \geq |X| \geq 2^{\iota}$ holds for every 
$j\in [\sigma_X + \g_{\iota-1}, \sigma_X + \g_{\iota}-1]$.
Since there is no active station woken up before $t_X$ we also have that
$|W_{i}(j)| = 0$  for $i>\iota$ and every $j\in [\sigma_X + \g_{\iota-1}, \sigma_X + \g_{\iota}-1]$.
By Lemma~\ref{iter}, there is a $\iota^{*}$-balanced interval $[\tau_1, \tau_2]$
for some $\iota^{*}\geq \iota$.

Let us assume that \eqref{up} does not hold, otherwise we are done. Let $j' \in [\tau_1, \tau_2]$ 
be such that $ \left| \bigcup_{i=1}^{\iota^*-1} W_{i}(j') \right| \geq 3\cdot |W_{\iota^*}(j')|$.
By Lemma~\ref{fat}, there exists an interval $[t_1, t_2]$ of size $\varphi(\iota^*+1)$ 
with $t_1 \geq j'$ such that
$|W_{\iota^*+2}(j)| \geq 2^{\iota^*+2}$. Letting $\omega = \iota^*+2$, we have found an interval of size 
$\varphi (\omega-1)$ such that $|W_{\omega}(j)| \geq 2^{\omega}$. 
By Lemma~\ref{iter}, there is an $\omega'$-balanced interval for some $\omega'\geq \omega$.
This process can be iterated until a balanced interval that satisfies condition \eqref{up} is identified.
\end{proof}

\begin{lemma}\label{C64}
There exists an $\omega$-light interval $[t_1, t_2]$, for some $1\leq \omega \leq \log k$.
\end{lemma}
\begin{proof}
Let $[t_1,t_2]$ be an $\omega$-balanced interval for some $1\leq \omega\leq \log k$, whose existence is guaranteed by Lemma~\ref{proper}. 
We can assume that every $j \in [t_1,t_2]$ satisfies condition~\eqref{up}, by this very Lemma.
Moreover, since the interval is $\omega$-balanced, we also have that $|W_{\omega}(j)| \leq 2^{\omega+2}$  for every $j \in [t_1,t_2]$, by condition (a) of Definition~\ref{bal}. 
This yields
\begin{eqnarray}\label{small}
   \Bigl| \bigcup_{i=1}^{\omega} W_i(j) \Bigr| &=& \Bigl| \bigcup_{i=1}^{\omega-1} W_i(j) \Bigr| + |W_{\omega}(j)| \nonumber \\
                                            &<& 3|W_{\omega}(j)| + |W_{\omega}(j)| \nonumber \\
                                            &\leq& 4\cdot 2^{\omega+2} = 2^{\omega+4},
\end{eqnarray}
for every $j \in [t_1,t_2]$.
Thus condition~(1) of Definition~\ref{d:light} is proved.

Next, we demonstrate condition~(2). 
By condition (a) of Definition~\ref{bal}, we have that $|W_{\omega}(j)| \geq 2^{\omega}$ for every 
$j \in [t_1, t_2]$. 
Therefore the following inequalities hold for every $j \in [t_1, t_2]$:
\[
          \Psi(j) \geq \frac{|W_{\omega}(j)|}{2^{\omega}} \geq 1
          \ .
\] 
It remains to prove that the upper bound of~\eqref{condition} holds for at least $\varphi(\omega-2)$ 
time steps.
Suppose, with the goal to arrive at a contradiction, that the number of time steps $j$ in $[t_1, t_2]$ that satisfies  the rightmost inequality of condition~\eqref{condition} is less than $\varphi(\omega-2)$.
Let $B \subseteq [t_1, t_2]$ be the set of balanced time steps $j \in [t_1, t_2]$ such that
condition~\eqref{condition} is not satisfied. 
By the assumption, the following inequalities hold:
\begin{equation}\label{contra}
  |B| > |[t_1,t_2]| - \varphi(\omega-2) \geq \frac{\varphi(\omega-2)}{2}.
\end{equation}
For any $j \in [t_1, t_2]$, let us consider
\[
U(j) = \bigcup_{i=1}^{\lg n} W_{i}(j) = \bigcup_{i=1}^{\omega} W_{i}(j)
\ ,
\] 
where the second identity follows by condition (b) of Definition~\ref{bal}.
We have that $|W_i(j)| = 0$ for $\omega < i\leq \lg n$ and for every $j \in [t_1, t_2]$, because $[t_1, t_2]$ is $\omega$-balanced.
Hence all stations in $W(j)$ lie on sections $i \leq \omega$ for every $j\in [t_1, t_2]$.
By the specification of sections, a station is in section~$i$, for $1\leq i\leq \omega$,  during $\varphi(i+1) - \varphi(i) \geq \varphi(i)$  time steps.
Therefore, for every $1\leq i\leq \omega$ 
\[
   \varphi(i)\max_{t_1 \leq j\leq t_2} |U(j)| \geq \sum_{j=t_1}^{t_2} |W_{i}(j)|
                                          \geq \sum_{j\in B} |W_{i}(j)|.
\]
We continue with the following estimates:
\begin{eqnarray*}
   \sum_{i=1}^{\omega} \max_{t_1 \leq j\leq t_2} |U(j)| &\geq& \sum_{i=1}^{\omega} \sum_{j\in B} \frac{|W_{i}(j)|}{\varphi(i)} \\
                      &  = & \sum_{j\in B} \sum_{i=1}^{\log k} \frac{|W_{i}(j)|}{\varphi(i)} \\
                      &  = & \frac{1}{c\log n} \sum_{j\in B} \sum_{i=1}^{\omega} \frac{|W_{i}(j)|}{2^i\cdot i^{1/b}} \\
                      &  \geq & \frac{1}{c\log n\log^{1/b}k} \sum_{j\in B} \sum_{i=1}^{\omega} \frac{|W_{i}(j)|}{2^i} \\  
                      &  > & \frac{1}{c\log n\log^{1/b}k}  \sum_{j\in B} 128\cdot\omega \;\;\mbox{ by the assumption}\\
                      &  = & \frac{|B|\cdot 128\cdot \omega}{c\log n\log^{1/b} k}.
\end{eqnarray*}
Therefore the following inequality holds:
\[
   \max_{t_1 \leq j\leq t_2} |U(j)| > \frac{|B|\cdot 128}{c\log n\log^{1/b} k}.
\]
By \eqref{contra} we obtain that the following estimate holds:
\[
   \max_{t_1 \leq j\leq t_2} |U(j)| > \frac{2^{\omega-3} c\log n\log^{1/b} k \cdot 128}{c\log n\log^{1/b} k} = 2^{\omega+4}.
\] 
This implies that there exists $j'\in [t_1, t_2]$ such that the following inequality holds
\[
  \Bigl| \bigcup_{i=1}^{\omega} W_{i}(j') \Bigr| > 2^{\omega+4},
\]
which contradicts \eqref{small}.
\end{proof}

\begin{lemma}
\label{isolating}

Let $\beta$ be a channel, for $1\leq \beta \leq b$. 
Let every station be executing a randomized algorithm as represented by a regular randomized transmission array.
Let $[t_1, t_2]$ be a light interval.
The probability that there exists a station $w\in W(t)$ that is $\beta$-isolated  at an arbitrary time step~$j$ such that $j\leq t$ and $t_1 \leq j\leq t_2$, is at least 
\[
\frac{\Psi(j)}{\log^{\beta/b}k}\cdot 4^{-\frac{\Psi(j)}{\log^{\beta/b}k}}.
\]
\end{lemma} 

\begin{proof}
Let $E_1(\beta, i,j)$ be the event ``there exists $w\in W_{i}(j)$ such that $T(\beta, w, j) = 1$'', and
let $E_2(\beta, i,j)$ be the event ``$T(u,\beta, j) = 0$ for all $l$ with $l\not = i$ and for every $u\in W_{l}(j)$.'' 
Let us say that $W(t)$ is $\beta$-isolated at time step $j\leq t$  if and only if there exists a station
$w\in W(t)$ that is $\beta$-isolated at time step $j$.
Clearly, $W(t)$ is $\beta$-isolated at time $j$ if and only if the following event occurs:
\[ 
    \bigcup_{i=1}^{\log n} \left(E_1(\beta, i,j) \cap E_2(\beta, i,j)\right)
    \ .
\] 
We use the following estimate on probability:
\begin{eqnarray*}
\Pr ( E_1(\beta, i,j) ) 
&\geq& 
\frac{|W_{i}(j)|}{2^{i}\lpbi} \left(1-\frac{1}{2^{i}\lpbi}\right)^{|W_{i}(j)|-1} \\
&\geq&
\frac{|W_{i}(j)|}{2^{i}\lpbi} \left( 1-\frac{1}{2^{i}\lpbi}\right)^{|W_{i}(j)|}
\end{eqnarray*}
and the following identity:
\[
\Pr (E_2(\beta, i,j) ) = \prod_{l=1,l \not = i}^{\log n} \left(1-\frac{1}{2^{l} l^{\beta/b} }\right)^{|W_{l}(j)|}.
\]
Events $E_1(\beta, i,j)$ and $E_2(\beta, i,j)$ are independent, so the following can be derived:
\begin{eqnarray*}
\Pr (E_1(\beta, i,j) \cap E_2(\beta, i,j) ) 
& \geq & 
\frac{|W_{i}(j)|}{2^{i}\lpbi} \prod_{l = 1}^{\log n} \left(1-\frac{1}{2^{l}l^{\beta/b}}\right)^{|W_{l}(j)|}\\
&=& 
\frac{|W_{i}(j)|}{2^{i}\lpbi} \prod_{l = 1}^{\log n} \left(1-\frac{1}{2^{l}l^{\beta/b}}\right)^{2^{l}l^{\beta/b} \frac{|W_{l}(j)|}{2^{l} l^{\beta/b}} }\\
&\geq& 
\frac{|W_{i}(j)|}{2^{i}\lpbi} \cdot {4}^{-\sum_{l = 1}^{\log n} \frac{|W_{l}(j)|}{2^{l} l^{\beta/b}} } 
  \ .
\end{eqnarray*}
The events $E_1(\beta, i,j) \cap E_2(\beta, i,j)$ are mutually exclusive, for any fixed $j$ and all $1 \leq i \leq \lg n$.
Additionally,  $W_i(j) = \emptyset$ for all $i > \log k$, as $[t_1, t_2]$ is a light interval.
Combining all this gives
\begin{eqnarray*}
\Pr (E_1(\beta, i,j) \cap E_2(\beta, i,j) ) 
&\geq& 
\sum_{l = 1}^{\log n} \frac{|W_{l}(j)|}{2^{l} l^{\beta/b}} \cdot {4}^{-\sum_{l = 1}^{\log n} \frac{|W_{l}(j)|}{2^{l} l^{\beta/b}} } \\
& = & 
\sum_{l = 1}^{\log k} \frac{|W_{l}(j)|}{2^{l} l^{\beta/b}} \cdot {4}^{-\sum_{l = 1}^{\log k} \frac{|W_{l}(j)|}{2^{l} l^{\beta/b}} }
\ . 				
\end{eqnarray*}
Observe that the function $x\cdot 4^{-x}$ is monotonically decreasing in~$x$. 
We apply this for 
\[
x = \sum_{l = 1}^{\log k} \frac{|W_{l}(j)|}{2^{l} l^{\beta/b}}
\ .
\]
Observe that the following inequality holds:
\[
\sum_{l = 1}^{\log k} \frac{|W_{l}(j)|}{2^{l} l^{\beta/b}} 
                            < \frac{1}{\log^{\beta/b}k }\sum_{l = 1}^{\log k} \frac{|W_{l}(j)|}{2^{l}}
                            \ .
\]
Combining these facts together justifies the following estimates
\begin{eqnarray*}
\Pr (E_1(\beta, i,j) \cap E_2(\beta, i,j) )
& \geq & 
\frac{1}{\log^{\beta/b}k }\sum_{l = 1}^{\log k} \frac{|W_{l}(j)|}{2^{l}} \cdot {4}^{-\frac{1}{\log^{\beta/bk}}\sum_{l = 1}^{\log k} \frac{|W_{l}(j)|}{2^{l}} } \\
& \geq & 
\frac{\Psi(j)}{\log^{\beta/b}k} \cdot 4^{\frac{\Psi(j)}{\log^{\beta/b}k}}
\ ,
\end{eqnarray*}
which completes the proof.
\end{proof}

\begin{lemma}
\label{probability}

Let every station be executing a randomized algorithm as represented by a regular randomized transmission array.
There exists an $\omega$-light interval $[t_1, t_2]$, for some $1\leq \omega \leq \log k$, that
contains at least $\varphi(\omega-2)$ time steps $j\in [t_1, t_2]$ such that 
the probability that there exists a station $w\in W(j)$ isolated at time $j$ is at least 
\[
         \frac{1}{4^{128}} \cdot \omega^{1/b} .
\]
\end{lemma}

\begin{proof}
By Lemma~\ref{C64}, there exists an $\omega$-light interval $[t_1, t_2]$, for some $1\leq \omega \leq \log k$.
There are at least $\varphi(\omega-2)$ time steps $j\in [t_1, t_2]$ with
$1 \leq \Psi(j) \leq 128\cdot \omega$. 
Let $T$ be the set of such time steps~$j$. 
We define sets $T_i$ as follows, for $1\le i\le b$:
\[
T_1 =  \{ j \in T |\ 1 \leq \Psi(j) \leq 128\cdot\omega^{1/b}  \}
\]
 and, for $q = 2, \ldots, b$,
\[
T_{q} = \{ j \in T |\ 128\cdot\omega^{(q-1)/b} < \Psi(j) \leq 128\cdot\omega^{q/b} \}\ .
\] 
It  suffices to show that for every time step $t\in T$ there exists a channel $\beta$, for $1\leq \beta \leq b$, such that the probability of $\beta$-isolating a station $w\in W(t)$ at time $t$ is at least 
\[
         \frac{1}{4^{128} \cdot \omega^{1/b} }
         \ .
\]
Let us consider a time step $t\in T$, so that $t \in T_q$ for some $1\leq q\leq b$.
By Lemma~\ref{isolating}, we have that if $t \in T_1$ then the probability that a station is $1$-isolated at a time step~$t$ is at least
\[
       \frac{\Psi(j)}{\lpbo}\cdot 4^{-\frac{\Psi(j)}{\lpbo}} > \frac{1}{\lpbo}\cdot 4^{-128\frac{\lpbo}{\lpbo}} 
       = \frac{1}{4^{128} \cdot \lpbo}
       \ .
\]
If $t\in T_\beta$, for $2\leq \beta \leq b$, then the probability that a station is $\beta$-isolated at time step $t$ is at least 
\[
       \frac{\Psi(j)}{\lpbo}\cdot 4^{-\frac{\Psi(j)}{\lpbo}} > \frac{128}{\lpbo}\cdot 4^{-\frac{\Psi(j)}{\lpbo}}
       = \frac{128}{4^{128} \cdot \lpbo}
       \ ,
\]
which completes the proof.
\end{proof}

\begin{lemma}
\label{l:t2}
Let $s$ be the time at which the first station wakes up and
let $[t_1,t_2]$ be an $\omega$-light interval, for some $\omega\le \log k$.
Then $t_2\le s + \varphi(\omega+1)$.
\end{lemma}

\begin{proof}
The interval $[t_1,t_2]$ is $\omega$-balanced, by Definition~\ref{d:light}.
We have that $W_j(t_2)$ is empty for every~$j$ such that $j>\omega$, by Definitions~\ref{baltime} and~\ref{bal}(b).
This means that no station is in a section bigger than~$\omega$, including those activated first at  time step~$s$.
Each station is activated after at most $\varphi(\omega+1)$ time steps because  the sum of the lengths of the first $i$ sections is $\g_i = \varphi(i+1)$.
\end{proof}

\begin{lemma}
\label{l:fraction}

Let $c$ in the definition of $\varphi$ be bigger than some sufficiently large constant.
There exists a waking array of length $2cn\log n\log^{1/b} k$ such that, for any transmission array,  there is an integer $0\le\omega\le\log k$ with the following  properties: 
\begin{enumerate}
\item[\rm (1)] 
There are at least $c\cdot 2^{\omega-259}\log n$ isolated positions by time $c\cdot 2^{\omega+1}\log n\log^{1/b} k$ .

\item[\rm (2)]
At least $c\cdot 2^{\omega-259}\log n$ isolated positions occur at time steps with at least $2^{\omega}$ but no more than~$2^{\omega+4}$ activated stations. 
\end{enumerate}
\end{lemma}

\begin{proof}
Consider a regular randomized transmission array, as defined in Definition~\ref{def:array}.
Assume also a sufficiently large $c>0$ in the definition of $\varphi(i)= c\cdot 2^{i} \log n\cdot i^{1/b}$, for any $1\le i\le \log n$. 
Consider an activation pattern, with the first activation at point zero.
By Lemma~\ref{probability}, there is $\omega\le \log k$ and an $\omega$-light interval $[t_1,t_2]$ such that there are at least $\varphi(\omega-2)$ time steps $j\in [t_1, t_2]$, where the probability that a station $w\in W(j)$ isolated at time step~$j$ exists is at least~$1/(4^{128}\lpbou)$.
We choose the smallest such $\omega$ and associate the corresponding $\omega$-light interval $[t_1,t_2]$ with the activation pattern.
Note that we can partition all activation patterns into disjoint classes based on the intervals associated with them.
The expected number of isolated positions in the $\omega$-light interval $[t_1,t_2]$ is at least 
\[
\varphi(\omega-2) \cdot \frac{1}{4^{128} \lpbou}
\geq
\varphi(\omega-2)\cdot \frac{1}{4^{128} (\omega-2)^{1/b} } \cdot\frac{(\omega-2)^{1/b}}{\lpbou}
\geq
c\cdot 2^{\omega -258} \log n
\ ,
\]
where we take $\omega\ge 3$.
By the Chernoff bound, the probability that the number of isolated positions is smaller than $c\cdot 2^{\omega -259} \log n$ is at most $\exp(-c\cdot 2^{\omega -261} \log n)$.

We want to apply the argument of the probabilistic method to the class of activation patterns associated with the $\omega$-light time interval $[t_1,t_2]$. 
To this end, we need an estimate from above of the number of all such activation patterns.
By Lemma~\ref{l:t2}, the rightmost end~$t_2$ of this time interval is not bigger than 
\[
\varphi(\omega+1) \le c\cdot 2^{\omega+1}\log n \log^{1/b} k\ .
\]
There are no more than  $2^{\omega+4}$ stations activated by time step~$t_2$, because $[t_1,t_2]$ is $\omega$-light. 
The number of different activation patterns in the class associated with the 
$\omega$-light interval $[t_1,t_2]$ is at most  $\binom{n}{2^{\omega+4}}(t_2)^{2^{\omega+4}}$.
This quantity can be estimated from above as
\begin{eqnarray*}
\left(\frac{ne}{2^{\omega+4}}\right)^{2^{\omega+4}}
\bigl(c\cdot 2^{\omega+1}\log n \log^{1/b} k \bigr)^{2^{\omega+4}} 
&=&
\exp\bigl(2^{\omega+4}\cdot \ln ((ce/8)\cdot n \log n \log^{1/b} k)\bigr)\\
&\le&
\exp\left(3\ln c \cdot 2^{\omega+4}\cdot \log n \right)
\ .
\end{eqnarray*}
This bound is smaller than $\exp(c\cdot 2^{\omega -261} \log n - 4\log(2cn \log n \log^{1/b} k))$ for a sufficiently large constant~$c$.
We combine the following two bounds:
\begin{itemize}
\item
this upper bound $\exp(c\cdot 2^{\omega -261} \log n - 4\log(2cn \log n \log^{1/b} k))$  on the number all activation patterns in the class associated with the $\omega$-light time interval $[t_1,t_2]$,  with
\item
the upper bound  $\exp(-c\cdot 2^{\omega -261} \log n)$  on the probability that for any fixed such activation pattern the number of isolated positions is smaller than $c\cdot 2^{\omega -259} \log n$.
\end{itemize}
We conclude that the probability of the event that there is an activation pattern associated with the $\omega$-light time interval $[t_1,t_2]$ with less than $c\cdot 2^{\omega -259} \log n$  isolated positions, is smaller than 
\[
\exp(c\cdot 2^{\omega -261} \log n - 4\log(2cn \log n \log^{1/b} k)) \cdot 
\exp(-c\cdot 2^{\omega -261} \log n)
=
\exp(- 4\log(2cn \log n \log^{1/b} k))
\ .
\]
Finally, observe that there are at most $2cn \log n \log^{1/b} k$ candidates for time step $t_1$
and also for $t_2$, by Lemma~\ref{l:t2} and the bound $\omega\le \log n$.
Hence, applying the union bound to the above events over all such feasible intervals, we obtain that the probability of the event that there is $\omega\le \log n$ and an activation pattern associated with some  $\omega$-light time interval $[t_1,t_2]$ with less than $c\cdot 2^{\omega -259} \log n$ isolated positions is smaller than 
\[
\exp(- 4\log(2cn \log n \log^{1/b} k)) \cdot  (2cn \log n \log^{1/b} k)^2 < 1/n^2 \le   1
\ .
\]
By the probabilistic-method argument, there is an instantiation of the random array, which is a regular array, for which the complementary event holds.

Note that more than the fraction  $1-1/n^2$ of random arrays defined in the beginning of the proof satisfy the complementary event.
Hence, this array satisfies Claim (1) with respect to any activation pattern.
Claim (2) follows by noticing that these occurrences of isolated positions take place in the corresponding $\omega$-light interval.
The interval, by definition, has no more than  $2^{\omega+4}$ stations activated by its end, and at least $2^{\omega}$ activated stations in the beginning.
This is because $\omega$-light interval is by definition an $\omega$-balanced interval, according to Definitions~\ref{bal}, \ref{baltime} and~\ref{d:light}.
\end{proof}

\Paragraph{Proof completed.}

We conclude with a proof of Theorem~\ref{thm:general-deterministic}.
There is an isolated position for every activation pattern by time $\cO(k\log n\log^{1/b} k)$.
This follows from point (1) of Lemma~\ref{l:fraction}.
To see this, notice that otherwise the $\omega$-light interval, which is also $\omega$-balanced, would have at least $2^\omega > k$ stations activated, by Definitions~\ref{bal} and~\ref{baltime}, contradicting the assumption.
Theorem~\ref{thm:general-deterministic} is thereby proved.

\Paragraph{Channels with random jamming.}

In the final part of this Section, we consider a model of a network in which channels may get jammed.
Assume that at each time step and on every channel a jamming error occurs with the probability~$p$, for $0\le p<1$, independently over time steps and channels. 
When a channel is jammed then the feedback it provides to the stations is the same as if there were a collision on this channel.
The case $p=0$ is covered by Theorem~\ref{thm:general-deterministic}.


\begin{theorem}
\label{thm:randomized-jamming}

For a given error probability $0< p<1$, there exists a waking array of  length $\cO(\log^{-1}(\frac{1}{p})\,n\log n\log^{1/b} k)$ providing wake-up in  $\cO(\log^{-1}(\frac{1}{p}) \, k\log n\log^{1/b} k)$ time, for any number $k\le n$ of spontaneously activated stations, with a probability that is at least $1 - 1/\text{\emph{poly}}(n)$.
\end{theorem}

\begin{proof} 
Let us set $c=c'\cdot\log^{-1}\frac{1}{p}$ for sufficiently large constant $c'$, and consider any activation pattern.
By Lemma~\ref{l:fraction}, at least $c\cdot 2^{\omega-259}\log n$ isolated positions occur by time $c\cdot 2^{\omega+1}\log^{1+1/b} n$ and by that time no more than $2^{\omega+4}$ stations are activated.
Each such an isolated position can be jammed independently with probability~$p$.
Therefore, the probability that all these positions are jammed, and thus no successful transmission 
occurs by time 
\[
c\cdot 2^{\omega+1}\log n\log^{1/b} k = \cO(\log^{-1}\Bigl(\frac{1}{p}\Bigr) \, k \log n\log^{1/b} k)\ ,
\]
is at most
\[
p^{c\cdot 2^{\omega-259}\log n}
=
\exp\left(c'\cdot \log^{-1}\Bigl(\frac{1}{p}\Bigr) \cdot 2^{\omega-259}\log n\cdot \ln p\right)
\ .
\]
This is smaller than $1/\mbox{poly}(n)$ for sufficiently large constant $c'$.
Here we use the fact that $\frac{\ln p}{\log(1/p)}$ is negative  for $0<p<1$.
When estimating the time of a successful wake-up, we relied on the fact that  $2^{\omega}$, which is the lower bound on the number of activated stations  by Lemma~\ref{l:fraction}(2), must be smaller than~$k$.
\end{proof}

\section{A Specialized Deterministic Algorithm}

\label{sec:large}

We give a deterministic algorithm that has a better time-performance bound than the one given in Theorem~\ref{thm:general-deterministic}.
The construction applies to networks with sufficiently many channels with respect to the number of nodes.
The main fact proved in this Section is as follows:


\begin{theorem}
\label{thm:many-channels}

If the numbers of channels~$b$ and nodes~$n$ satisfy $b>\lbl$ then there exists a deterministic waking array $\cT$ of length $\cO(\frac{n}{b}\log n\log(b\log n))$ which, when used to instantiate the generic algorithm \textsc{Wake-Up}, produces an algorithm \textsc{Wake-Up}\,$(\cT)$ that wakes up the network in  time $\cO(\frac{k}{b}\log n\log(b\log n))$, for up to $k\le n$ stations  activated spontaneously.
\end{theorem}

The proof is by way of showing the existence of a waking array, as defined in Definition~\ref{schedule}, for a section length defined as $\varphi(i)=c\cdot (2^i/b) \lg n\lbl$. 
Note that Lemmas~\ref{zigzag} to~\ref{proper} as well as Lemma~\ref{l:t2}  hold for the current specification of function $\varphi$, as their proofs do not refer to the value of this function.
The following Lemma corresponds to Lemma~\ref{C64}, which was proved for $\varphi(i)= c\cdot 2^i\cdot i^{1/b} \log n$, while now we prove an analogous statement for $\varphi(i)=c\cdot (2^i/b) \lg n\lbl$.

\begin{lemma}
\label{C64-large}
There exists an $\omega$-light interval $[t_1, t_2]$, for some $1\leq \omega \leq \log n$.
\end{lemma}

\begin{proof}
Let $[t_1,t_2]$ be an $\omega$-balanced interval, which exists by Lemma~\ref{proper}. 
By that very Lemma, we can assume that every $j \in [t_1,t_2]$ satisfies condition~\eqref{up}.
Moreover, since the interval is $\omega$-balanced, we also have that $|W_{\omega}(j)| \leq 2^{\omega+2}$ for every $j \in [t_1,t_2]$, by condition (a) of Definition~\ref{bal}. 
We conclude with the following upper bound, for every $j \in [t_1,t_2]$:
\begin{eqnarray}\label{small-large}
   \Bigl| \bigcup_{i=1}^{\omega} W_i(j) \Bigr| &=& \Bigl| \bigcup_{i=1}^{\omega-1} W_i(j) \Bigr| + |W_{\omega}(j)| \nonumber \\
                                            &<& 3|W_{\omega}(j)| + |W_{\omega}(j)| \nonumber \\
                                            &\leq& 4\cdot 2^{\omega+2} = 2^{\omega+4}
                                            \ .
\end{eqnarray}
This proves condition~(1) of Definition~\ref{d:light}.

Next we prove condition~(2).
By condition~(a) of Definition~\ref{bal}, we know that $|W_{\omega}(j)| \geq 2^{\omega}$ for every 
$j \in [t_1, t_2]$. 
Therefore, the following bounds hold for every $j \in [t_1, t_2]$:
\[
          \Psi(j) \geq \frac{|W_{\omega}(j)|}{2^{\omega}} \geq 1
          \ .
\] 

What remains to show is the upper bound of~\eqref{newcondition}.
Suppose, to arrive at a contradiction, that the number of time steps $j$ in $[t_1, t_2]$, that satisfies 
the rightmost inequality of condition~\eqref{newcondition}, is less than $\varphi(\omega-2)$. 
Let $B \subseteq [t_1, t_2]$ be the set of balanced time steps $j \in [t_1, t_2]$ such that
condition~\eqref{newcondition} is not satisfied. 
By the assumption, the following is the case:
\begin{equation}\label{contra-large}
  |B| > |[t_1,t_2]| - \varphi(\omega-2) = \frac{\varphi(\omega-2)}{2}.
\end{equation}
For any $j \in [t_1, t_2]$, let us consider
\[
U(j) = \bigcup_{i=1}^{\lg n} W_{i}(j) = \bigcup_{i=1}^{\omega} W_{i}(j)
\ ,
\] 
where the second identity follows by condition~(b) of Definition~\ref{bal}.
By the specification of the array, any station belongs to section~$i$ during $\varphi(i+1) - \varphi(i) \geq \varphi(i)$ time steps, for $1\leq i\leq \lg n$.
Therefore we have the following bounds
\[
   \varphi(i)\max_{t_1 \leq j\leq t_2} |U(j)| \geq \sum_{j=t_1}^{t_2} |W_{i}(j)|
                                          \geq \sum_{j\in B} |W_{i}(j)|
\]
for every $1\leq i\leq \lg n$.
This in turn allows to obtain the following bound:
\begin{eqnarray*}
  \sum_{i=1}^{\lg n} \max_{t_1 \leq j\leq t_2} |U(j)| &\geq& \sum_{i=1}^{\lg n} \sum_{j\in B} \frac{|W_{i}(j)|}{\varphi(i)} \\
                                                   &  = & \sum_{j\in B} \sum_{i=1}^{\lg n} \frac{|W_{i}(j)|}{\varphi(i)} \\
                                                   &  = & \frac{1}{c(\lg n\lbl)/b} \sum_{j\in B} \sum_{i=1}^{\lg n} \frac{|W_{i}(j)|}{2^i} \\
                                                  &  > & \frac{1}{c(\lg n\lbl)/b}  \sum_{j\in B} 128\cdot\log n \\
                                                   &  = & \frac{128 b |B|}{c\lbl} \ .
\end{eqnarray*}
This gives the following estimate:
\[
   \max_{t_1 \leq j\leq t_2} |U(j)| > \frac{128 b |B|}{c\lg n\lbl} \ .
\]
By applying~\eqref{contra-large}, we obtain
\[
   \max_{t_1 \leq j\leq t_2} |U(j)| > \frac{128 b \cdot c (2^{\omega-3}/b)\lg n\lbl}{c\lg n\lbl} = 2^{\omega+4} \ .
\] 
This implies that there exists $j'\in [t_1, t_2]$ such that 
\[
  \Bigl| \bigcup_{i=1}^{\omega} W_{i}(j') \Bigr| > 2^{\omega+4} \ ,
\]
which contradicts \eqref{small-large}.
\end{proof}

\begin{lemma}
\label{isolating-large}

Let $\beta$ be a channel, for $1\leq \beta \leq b$. 
Let every station be executing the randomized algorithm as represented by a modified randomized transmission array.
The probability that there exists a station $w\in W(t)$ that is $\beta$-isolated  at any time step $j\leq t$ is at least 
\[
\Psi(j)\cdot b\cdot 2^{-\beta^*} \cdot 4^{-\Psi(j)\cdot b\cdot 2^{-\beta^*}} \ .
\]
\end{lemma} 

\begin{proof}
Let $E_1(\beta, i,j)$ be the event ``there exists $w\in W_{i}(j)$ such that $T(\beta, w, j) = 1$'', and
let $E_2(\beta, i,j)$ be the event ``$T(u,\beta, j) = 0$ for all $l$ with $l\not = i$ and for every $u\in W_{l}(j)$.'' 
Let us say that $W(t)$ is $\beta$-isolated at time step $j\leq t$  if and only if there exists a station
$w\in W(t)$ that is $\beta$-isolated at time step $j$.
Clearly, $W(t)$ is $\beta$-isolated at time $j$ if and only if the following event occurs:
\[ 
    \bigcup_{i=1}^{\log n} \bigl(E_1(\beta, i,j) \cap E_2(\beta, i,j)\bigr)
    \ .
\] 
We use the following inequalities
\begin{eqnarray*}
\Pr (E_1(\beta, i,j) )
&\geq&
 |W_{i}(j)|\cdot b \cdot 2^{-i -\beta^*} \Bigl(1-b \cdot 2^{-i -\beta^*}\Bigr)^{|W_{i}(j)|-1} \\
 &\geq &
 |W_{i}(j)|\cdot b \cdot 2^{-i -\beta^*} \Bigl(1-b \cdot 2^{-i -\beta^*}\Bigr)^{|W_{i}(j)|}  
\end{eqnarray*}
combined with the following identity
\[
\Pr (E_2(\beta, i,j) ) = \prod_{l=1,l \not = i}^{\log n} \left(1-b \cdot 2^{-l -\beta^*}\right)^{|W_{l}(j)|} \ .
\]
Events $E_1(\beta, i,j)$ and $E_2(\beta, i,j)$ are  independent. 
It follows that
\begin{eqnarray*}
\Pr (E_1(\beta, i,j) \cap E_2(\beta, i,j) ) 
& \geq & 
		|W_{i}(j)|\cdot b \cdot 2^{-i -\beta^*}  \prod_{l = 1}^{\log n}
				\left(1-b \cdot 2^{-l -\beta^*}\right)^{|W_{l}(j)|}
		 \\
		        &=& |W_{i}(j)|\cdot b \cdot 2^{-i -\beta^*}  \prod_{l = 1}^{\log n}
			\left(1-b \cdot 2^{-l -\beta^*}\right)^{(2^{l +\beta^*}/b) \cdot b|W_{l}(j)|2^{-l-\beta^*}}
			\\
		        &\geq& |W_{i}(j)|\cdot b \cdot 2^{-i -\beta^*}  \cdot
			      {4}^{-\sum_{l = 1}^{\log n} (b|W_{l}(j)|2^{-l-\beta^*})} \\
			      &=& \frac{|W_{i}(j)|}{2^{i}}\cdot (b\cdot 2^{-\beta^*}) \cdot{4}^{-\Psi(j)\cdot (b\cdot 2^{-\beta^*})} \ .
\end{eqnarray*}
To conclude, observe that the events $E_1(\beta, i,j) \cap E_2(\beta, i,j)$ are mutually exclusive, for any~$j$ and $1 \leq i \leq \lg n$.
\end{proof}

\begin{lemma}
\label{probability-large}

Let every station be executing the randomized algorithm as represented by a modified randomized transmission array.
There exists an $\omega$-light interval $[t_1, t_2]$, for some $1\leq \omega \leq \log n$, that
contains at least $\varphi(\omega-2)$ time steps $j\in [t_1, t_2]$ such that  in each of these steps the number of channels~$\beta$, with the probability of a $\beta$-isolated station being at least $1/8$, is at least $\bigl\lfloor\frac{b}{\lbl}\bigr\rfloor$.
\end{lemma}

\begin{proof}
By Lemma~\ref{C64-large}, there are at least $\varphi(\omega-2)$ time steps $j\in [t_1, t_2]$ such 
that the inequalities $1 \leq \Psi(j) \leq 128\cdot \lg n$ hold. 
Let $T$ be the set of such time steps.
Let us define a partition of $T$ into sets 
\[
T_{q} = \{ j \in T |\ 2^{q} < b\cdot \Psi(j) \leq 2^{q+1} \}
\ , 
\]
for $0\le q< \lbl$. 
It suffices to show that for every time step $t\in T$ there exists $\bigl\lfloor\frac{b}{\lbl}\bigr\rfloor$ channels $\beta$, for $1\leq \beta \leq b$, such that the probability of $\beta$-isolating a station $w\in W(t)$ at time $t$ is at least~$1/8$.

Let us take any time step $t\in T$, so that $t \in T_q$ for some $0\leq q<\lbl$.
By Lemma~\ref{isolating-large}, if $t \in T_q$ then for each $\beta$ such that $\beta^*=q$,
the probability that a station is $\beta$-isolated at time step~$t$ is at least
\[
\Psi(j)\cdot b\cdot 2^{-\beta^*} \cdot 4^{-\Psi(j)\cdot b\cdot 2^{-\beta^*}} 
\ge
\left(2^{q+1}\cdot 2^{-q}\right) \cdot 4^{-2^{q+1}\cdot 2^{-q}} 
=
1/8
\ ,
\]
where we use the fact that the function $x\cdot 4^{-x}$ is monotonically decreasing in~$x$.
To conclude, notice that there are at least $\bigl\lfloor\frac{b}{\lbl}\bigr\rfloor$ channels~$\beta$ satisfying $\beta^*=q$, for any given $q$ such that $0\le q< \lbl$.
\end{proof}

\begin{lemma}
\label{l:fraction-large}
Let $c$ in the definition of $\varphi$ be bigger than some sufficiently large constant.
There exists a waking array  of length $\frac{2cn}{b}\log n\lbl$
such that for any activation pattern, there is an integer $0\le\omega\le\log n$ with the following properties: 
\begin{enumerate}
\item[\rm (1)] 

There are at least $c\cdot 2^{\omega-6}\log n$ isolated positions by time step $c\cdot (2^{\omega+1}/b)\log n\lbl$,

\item[\rm(2)] 
These positions occur at time step with at least $2^{\omega}$ but no more than $2^{\omega+4}$ activated stations. 
\end{enumerate}
\end{lemma}

\begin{proof}
Let us consider a modified randomized  transmission array.
Let us assume that $c>0$    in the specification 
\[
\varphi(i)= c\cdot (2^{i}/b) \log n \lbl\ 
\]
is sufficiently large, for any $1\le i\le \log n$. 
Observe that the length of the schedules is not bigger than 
\[
\varphi(i+1) \le \frac{2cn}{b}\log n\lbl\ .
\]
Let us consider an activation pattern with the first activation at time step~$0$.
By Lemma~\ref{probability-large}, there is $\omega\le \log n$ and $\omega$-light interval $[t_1,t_2]$ such that
there are at least $\varphi(\omega-2)$ time steps $j\in [t_1, t_2]$ such that each of them has
at least $\Bigl\lfloor\frac{b}{\lbl}\Bigr\rfloor$ channels $\beta$ with
the probability of $\beta$-isolation of a station $w\in W(j)$ being at least $1/8$.
We choose the smallest such an~$\omega$ and associate the corresponding $\omega$-light interval $[t_1,t_2]$
with the activation pattern.
We can clearly partition all activation patterns into disjoint classes based on the intervals associated
with them.

Observe that the expected number of isolated positions in the $\omega$-light interval $[t_1,t_2]$ is at least 
\[
\varphi(\omega-2) \cdot \left\lfloor\frac{b}{\lbl}\right\rfloor \cdot \frac{1}{8}
=
c\cdot 2^{\omega -5} \log n
\ .
\]
By the Chernoff bound, the probability that the number of isolated positions is smaller than
$c\cdot 2^{\omega -6} \log n$ is at most $\exp(-c\cdot 2^{\omega -8} \log n)$.

In order to apply the  probabilistic-method argument to the class of activation patterns associated with
the $\omega$-light time interval $[t_1,t_2]$, it remains to estimate from above the number of all such activation patterns.
By Lemma~\ref{l:t2}, the rightmost end~$t_2$ of this time interval is not bigger than 
\[
\sum_{i=1}^\omega \varphi(i) \le c\cdot (2^{\omega+1}/b)\log n\llbl\  .
\]
Next observe that since $[t_1,t_2]$ is $\omega$-light, there are no more than 
$2^{\omega+4}$ stations activated by time step~$t_2$. 
Hence, the number of different activation patterns in the class associated with the 
$\omega$-light interval $[t_1,t_2]$
is at most 
\begin{eqnarray*}
{n \choose 2^{\omega+4}} \left(t_2\right)^{2^{\omega+4}} 
&\le&
\left( \frac{ne}{2^{\omega+4}}\right)^{2^{\omega+4}} \left(c\cdot (2^{\omega+1}/b)\log n \llbl \right)^{2^{\omega+4}} \\
&=&
\exp\left(2^{\omega+4}\cdot \ln ((ce/8)\cdot (n/b))\log n\llbl)\right)\\
&\le&
\exp\left(3\ln c \cdot 2^{\omega+4}\cdot \log n\right)
\ ,
\end{eqnarray*}
which is smaller than $\exp(c\cdot 2^{\omega -8} \log n - 4\log(2cn\log n))$, for a sufficiently large constant~$c$.
Next we combine the following two bounds:
\begin{itemize}
\item
the upper bound $\exp(c\cdot 2^{\omega -8} \log n - 4\log(2cn\log n))$ on the number all activation patterns in the class associated with the $\omega$-light time interval $[t_1,t_2]$, with
\item
the upper bound $\exp(-c\cdot 2^{\omega -8} \log n)$ on the probability that for any fixed such an activation pattern the number of isolated positions is smaller than $c\cdot 2^{\omega -6} \log n$.
\end{itemize}
This allows to conclude that the probability of the event that there is an activation pattern associated with the $\omega$-light time interval $[t_1,t_2]$ with less than $c\cdot 2^{\omega -6} \log n$ isolated positions is smaller than 
\[
\exp(c\cdot 2^{\omega -8} \log n - 4\log(2cn\log n)) \cdot 
\exp(-c\cdot 2^{\omega -8} \log n)
=
\exp(- 4\log(2cn\log n))
\ .
\]
There are at most $2cn\log n$ candidates for time step $t_1$ and also for $t_2$, by Lemma~\ref{l:t2} applied to $\varphi(\omega)=c(2^\omega/b)\log n\lbl$ 
and the bound $\omega\le \log n$.
We apply the union bound to these events over all such feasible intervals.
This gives that the probability that there is $\omega\le \log n$ and an activation pattern associated with some  $\omega$-light time interval $[t_1,t_2]$ with less than $c\cdot 2^{\omega -6} \log n$ isolated positions  is smaller than 
\[
\exp(- 4\log(2cn\log n))
\cdot 
(2cn\log n)^2
<
1/n^2 \leq 1
\ .
\]
Thus, by the probabilistic-method argument, there is an instantiation of the random array, which is a deterministic array, for which the complementary event holds.
Hence, this array satisfies claim~(1) of this Lemma with respect to any activation pattern.
Claim~(2) follows when one observes  that these occurrences of isolated positions
take place in the corresponding $\omega$-light interval, which by definition has no more than 
$2^{\omega+4}$ stations activated by its end, and at least $2^{\omega}$ activated stations
in its beginning.
This is because $\omega$-light interval is by definition an $\omega$-balanced interval,
according to Definitions~\ref{bal}, \ref{baltime} and~\ref{d:light}.
\end{proof}

\Paragraph{Proof completed.}

We conclude with the proof of Theorem~\ref{thm:many-channels}.
There is an isolated position by time $\cO(\frac{k}{b}\log n\llbl)$  for every activation pattern.
This follows from point (1) of Lemma~\ref{l:fraction-large}.
Indeed, otherwise the $\omega$-light interval, which is also $\omega$-balanced, would have
at least $2^\omega > k$ stations activated, by Definitions~\ref{bal} and~\ref{baltime}, contrary to the assumptions.

\Paragraph{Channels with random jamming.}

We also consider a model of random jamming of channels for the case of sufficiently many channels.
Let us assume that at each time step and on every channel, a jamming error occurs with the probability~$p$, where $0\le p<1$, independently over time steps and channels. 
The case $p=0$ is covered by Theorem~\ref{thm:general-deterministic}.

\begin{theorem}
\label{matrix-failures-large}

For a given error probability~$p$, where $0< p<1$, if the numbers of channels~$b$ and nodes~$n$ satisfy the inequality $b>\lbl$, then there exists a waking array of  length $\cO(\log^{-1}(\frac{1}{p})\, \frac{n}{b}\log n\log(b\log n))$ providing wake-up in  time $\cO(\log^{-1}(\frac{1}{p}) \, \frac{k}{b} \log n\log(b\log n))$, for any number $k\le n$ of spontaneously activated stations, with a probability that is at least $1 - 1/\text{\emph{poly}}(n)$.
\end{theorem}

\begin{proof} 
Let us set $c=c'\cdot\log^{-1}\frac{1}{p}$, for a sufficiently large constant $c'$, and consider any activation pattern.
By Lemma~\ref{l:fraction-large}, $c\cdot 2^{\omega-6}\log n$ isolated positions occur by time step $c\cdot (2^{\omega+1}/b)\log n\llbl$ and by that time step no more than $2^{\omega+4}$ stations get  activated.
Each such an isolated position is jammed independently with probability~$p$.
Therefore, the probability that all these positions are jammed, and thus no successful transmission 
occurs by time 
\[
c\cdot (2^{\omega+1}/b)\log n\llbl = \cO(\log^{-1}\Bigl(\frac{1}{p}\Bigl) \, \frac{k}{b}\log n\log(b\log n))
\ ,
\]
is at least 
\[
p^{c\cdot 2^{\omega-6}\log n}
=
\exp\left(c'\cdot \log^{-1}\Bigl(\frac{1}{p}\Bigr) \cdot (2^{\omega-6}/b)\log n\cdot \ln p\right)
\ ,
\]
which is smaller than $1/\mbox{poly}(n)$ for sufficiently large constant $c'$.
Here we use the fact that $\frac{\ln p}{\log(1/p)}$ is negative for $0<p<1$.
When bounding  the time step of a successful wake-up to occur, we rely on the fact that  $2^{\omega}$, which is the lower bound on the number of activated stations  by Lemma~\ref{l:fraction-large}(2), must be smaller than~$k$.
\end{proof}

\section{Conclusion}

We considered waking up a multi-channel single-hop radio network by deterministic and randomized algorithms.
To assess optimality of a solution, we gave a lower bound $\frac{k}{4b}\lg \frac{n}{k} - \frac{k+1}{b}$ on time of a deterministic algorithm, which holds when both $k$ and $n$ are known.

This lower bound can be beaten by randomized algorithms when $k$ is known, as we demonstrated that a randomized algorithm exists that refers to $k$ and works in time $\cO(k^{1/b}\ln \frac{1}{\epsilon})$ with a large probability. 
This shows a separation between the best performance bounds of randomized and deterministic wake-up algorithms when the parameter~$k$ is known, even for just two channels.

We may interpret the parameters $k$ and $b$ as representing scalability of an algorithmic solution, by the presence of factors $k$ and $1/b$ in time-performance bounds.
This could mean that an algorithm that scales perfectly with $k$ and $b$ has time performance of the form $\cO(\frac{k}{b} \cdot f(n,b,k))$, for some function~$f(n,b,k)$ such that $f(n_0,b,k)=\cO(1)$ for any constant $n_0$ and the variables $b$ and $k$ growing unbounded.

Deterministic algorithms given in this paper are developed for the case when $n$ is known but $k$ is unknown.
Our general solution operates in time $\cO(k\log^{1/b} k\log n)$.
This means that $k\log^{1/b} k$ reflects scalability with $k$, which is close to linear in~$k$, while the scalability with~$b$ is poor, as $1/b$ is not a factor in the performance bound at all.
When sufficiently many channels are available, we show that  a multi-channel can be woken up deterministically in time $\cO(\frac{k}{b}\log n\log(b\log n))$. 
The respective algorithm is effective in two ways.
The first one is about time performance: the algorithm misses time optimality by at most a  poly-logarithmic factor that is $\cO(\log n(\log b +\log\log n))$, because of the lower bound~$\frac{k}{4b}\lg \frac{n}{k} - \frac{k+1}{b}$.
The second one is about scalability: the algorithm scales perfectly with the unknown~$k$, and also its scalability with~$b$ is~$\frac{\log b}{b}$, so it misses optimality in that respect by the factor of~$\log b$ only. 

\noindent
\textbf{Acknowledgement:} 
We want to thank the anonymous reviewers, of a manuscript that resulted in publishing~\cite{ChlebusDK16}, for comments that led to improving the submission; in particular, for pointing out that the lower bound from Clementi et al.~\cite{ClementiMS03} could be used to structure a short argument for a lower bound for multi channel networks.

\bibliographystyle{abbrv}
\bibliography{multi-channel-wakeup}

\end{document}